\documentclass[a4paper,cleveref, autoref, thm-restate,svgnames]{lipics-v2019}
\hypersetup{
    colorlinks=true,
    linkcolor=DarkSlateBlue,
    filecolor=magenta, 
    citecolor=SeaGreen,
    urlcolor=magenta,
    pdftitle={Overleaf Example},
    pdfpagemode=FullScreen,
    }

\hideLIPIcs{}
\allowdisplaybreaks
\title{On the Parameterized Complexity of Diverse SAT}

\titlerunning{On the Parameterized Complexity of Diverse SAT} 

\author{Neeldhara Misra}{Department of Computer Science and Engineering \and Indian Institute of Technology Gandhinagar \and \href{https://www.neeldhara.com}{https://www.neeldhara.com} }{neeldhara.m@iitgn.ac.in}{}{Supported by IIT Gandhinagar.}

\author{Harshil Mittal}{Department of Computer Science and Engineering \and Indian Institute of Technology Gandhinagar}{mittal\_harshil@iitgn.ac.in}{}{Supported by IIT Gandhinagar.}

\author{Ashutosh Rai}{Department of Mathematics \and Indian Institute of Technology Delhi\and \href{https://web.iitd.ac.in/~raiashutosh}{https://web.iitd.ac.in/~raiashutosh}}{ashutosh.rai@maths.iitd.ac.in}{}{Supported by IIT Delhi.}

\authorrunning{N. Misra, H. Mittal and A. Rai} 

\Copyright{John Q. Public and Joan R. Public} 


\category{} 

\relatedversion{} 

\supplement{}

\funding{}

\acknowledgements{We thank Saraswati Girish Nanoti for helpful discussions.}

\nolinenumbers 

\EventEditors{John Q. Open and Joan R. Access}
\EventNoEds{2}
\EventLongTitle{42nd Conference on Very Important Topics (CVIT 2016)}
\EventShortTitle{CVIT 2016}
\EventAcronym{CVIT}
\EventYear{2016}
\EventDate{December 24--27, 2016}
\EventLocation{Little Whinging, United Kingdom}
\EventLogo{}
\SeriesVolume{42}
\ArticleNo{23}
\usepackage{amsthm}
\usepackage{graphicx}
\usepackage{mathtools}
\usepackage[usestackEOL]{stackengine}
\usepackage{cancel}
\usepackage{charter}
\usepackage{relsize}

\newcommand{\Oplus}{\ensuremath{\vcenter{\hbox{\scalebox{1.5}{$\oplus$}}}}}

\usepackage[backgroundcolor=SkyBlue,linecolor=SkyBlue,obeyFinal]{todonotes}

\makeatletter
\providecommand\@dotsep{5}
\def\listtodoname{}
\def\listoftodos{\@starttoc{tdo}\listtodoname}
\makeatother

\newcounter{nmcomment}

\usepackage{parskip}

\usepackage{comment}
\excludecomment{shortver}
\includecomment{longver}

\usepackage{wrapfig}
\usepackage{tikz}
\usetikzlibrary{positioning}
\usetikzlibrary{arrows}
\usetikzlibrary{decorations.pathmorphing}
\usetikzlibrary{decorations.markings}

\usepackage{latexsym,amssymb,amsmath,mathtools,eulervm,xspace, bbding}

\usepackage{enumitem}
\usepackage{appendix}
\usepackage{apxproof}

\newcommand{\FPT}{\ensuremath{\mathsf{FPT}}\xspace}
\newcommand{\WOH}{\ensuremath{\mathsf{W}[1]}-hard\xspace}

\newcommand{\NP}{\ensuremath{\mathsf{NP}}\xspace}
\newcommand{\NPH}{\ensuremath{\mathsf{NP}}-hard\xspace}
\newcommand{\NPC}{\ensuremath{\mathsf{NP}}-complete\xspace}

\begin{document}

\maketitle
Note: This is full version of the corresponding ISAAC 2024 paper \cite{misra2024parameterized}.
\begin{abstract}
We study the \textsc{Boolean Satisfiability problem (SAT)} in the framework of diversity, where one asks for multiple solutions that are mutually far apart (i.e., sufficiently dissimilar from each other) for a suitable notion of distance/dissimilarity between solutions. Interpreting assignments as bit vectors, we take their Hamming distance to quantify dissimilarity, and we focus on the problem of finding two solutions. Specifically, we define the problem \textsc{Max}  \textsc{Differ SAT} (resp. \textsc{Exact Differ SAT}) as follows: Given a Boolean formula $\phi$ on $n$ variables, decide whether $\phi$ has two satisfying assignments that differ on at least (resp. exactly) $d$ variables. We study the classical and parameterized (in parameters $d$ and $n-d$) complexities of \textsc{Max Differ SAT} and \textsc{Exact Differ SAT}, when restricted to some classes of formulas on which SAT is known to be polynomial-time solvable. In particular, we consider affine formulas, Krom formulas (i.e., $2$-CNF formulas) and hitting formulas.

For affine formulas, we show the following: Both problems are polynomial-time solvable when each equation has at most two variables. \textsc{Exact Differ SAT} is \NPH, even when each equation has at most three variables and each variable appears in at most four equations. Also, \textsc{Max Differ SAT} is \NPH, even when each equation has at most four variables. Both problems are \WOH in the parameter $n-d$. In contrast, when parameterized by $d$, \textsc{Exact Differ SAT} is \WOH, but \textsc{Max Differ SAT} admits a single-exponential \FPT algorithm and a polynomial-kernel. 

For Krom formulas, we show the following: Both problems are polynomial-time solvable when each variable appears in at most two clauses. Also, both problems are \WOH in the parameter $d$ (and therefore, it turns out, also \NPH), even on monotone inputs (i.e., formulas with no negative literals). Finally, for hitting formulas, we show that both problems can be solved in polynomial-time.
\end{abstract}

\section{Introduction}
We initiate a study of the problem of finding two satisfying assignments to an instance of SAT, with the goal of maximizing the number of variables that have different truth values under the two assignments, in the parameterized setting. This question is motivated by the broader framework of finding ``diverse solutions'' to optimization problems. When a real-world problem is modelled as a computational problem, some contextual side-information is often lost. So, while two solutions may be equally good for the theoretical formulation, one of them may be better than the other for the actual practical application. A natural fix is to provide multiple solutions (instead of just one solution) to the user, who may then pick the solution that best fulfills her/his need. However, if the solutions so provided are all quite similar to each other, they may exhibit almost identical behaviours when judged on the basis of any relevant external factor. Thus, to ensure that the user is able to meaningfully compare the given solutions and hand-pick one of them, she/he must be provided a collection of \emph{diverse solutions}, i.e., a few solutions that are sufficiently dissimilar from each other. This framework of \emph{diversity} was proposed by Baste et. al. \cite{baste2022diversity}. Since the late 2010s, several graph-theoretic and matching problems have been studied in this setting from an algorithmic standpoint. These include diverse variants of vertex cover \cite{baste2019fpt}, feedback vertex set \cite{baste2019fpt}, hitting set \cite{baste2019fpt}, perfect/maximum matching \cite{fomin2020diverse}, stable matching \cite{ganesh2021disjoint}, weighted basis of matroid \cite{fomin2023diverse}, weighted common independent set of matroids \cite{fomin2023diverse}, minimum $s$-$t$ cut \cite{de2023finding}, spanning tree \cite{hanaka2021finding} and non-crossing matching \cite{misra2022diverse}.

The \textsc{Boolean Satisfiability problem} (SAT) asks whether a given Boolean formula has a satisfying assignment. This problem serves a crucial role in complexity theory \cite{karp2021reducibility}, cryptography \cite{mironov2006applications} and artificial intelligence \cite{vizel2015boolean}. In the early 1970s, SAT became the first problem proved to be \NPC in independent works of Cook \cite{cook2023complexity} and Levin \cite{levin1973universal}. Around the same time, Karp \cite{karp2021reducibility} built upon this result by showing \NP-completeness of twenty-one graph-theoretic and combinatorial problems via reductions from SAT. In the late 1970s, Schaefer \cite{schaefer1978complexity} formulated the closely related  \textsc{Generalized Satisfiability problem} (SAT($S$)), where each constraint applies on some variables, and it forces the corresponding tuple of their truth-values to belong to a certain Boolean relation from a fixed finite set $S$. His celebrated dichotomy result listed six conditions such that SAT($S$) is  polynomial-time solvable if $S$ meets one of them; otherwise, SAT($S$) is \NPC. 

Since SAT is \NPC, it is unlikely to admit a polynomial-time algorithm, unless $\mathsf{P} = \NP$. Further, in the late 1990s, Impaglliazo and Paturi \cite{impagliazzo2001complexity} conjectured that SAT is unlikely to admit even sub-exponential time algorithms, often referred to as the \emph{exponential-time hypothesis}. To cope with the widely believed hardness of SAT, several special classes of Boolean formulas have been identified for which SAT is polynomial-time solvable. In the late 1960s, Krom \cite{krom1967decision} devised a quadratic-time algorithm to solve SAT on 2-CNF formulas. In the late 1970s, follow-up works of Even et. al. \cite{even1975complexity} and Aspvall et. al. \cite{aspvall1979linear} proposed linear-time algorithms to solve SAT on $2$-CNF formulas. These algorithms used limited back-tracking and analysis of the strongly-connected components of the implication graph respectively. In the late 1980s, Iwama \cite{iwama1989cnf} introduced the class of hitting formulas, for which he gave a closed-form expression to count the number of satisfying assignments in polynomial-time. It is also known that SAT can be solved in polynomial-time on affine formulas using Gaussian elimination \cite{grcar2011ordinary}. Some other polynomial-time recognizable classes of formulas for which SAT is polynomial-time solvable include Horn formulas \cite{dowling1984linear, scutella1990note}, CC-balanced formulas \cite{conforti1994balanced}, matched formulas \cite{franco2003perspective}, renamable-Horn formulas \cite{lewis1978renaming} and $q$-Horn DNF formulas \cite{boros1990polynomial, boros1994recognition}.

\textbf{Diverse variant of SAT}. In this paper, we undertake a complexity-theoretic study of SAT in the framework of diversity. We focus on the problem of finding a \emph{diverse pair of satisfying assignments} of a given Boolean formula, and we take the number of variables on which the two assignments differ as a measure of dissimilarity between them. Specifically, we define the problem \textsc{Max} \textsc{Differ} SAT (resp. \textsc{Exact} \textsc{Differ} SAT) as follows: Given a Boolean formula $\phi$ on $n$ variables and a non-negative integer $d$, decide whether there are two satisfying assignments of $\phi$ that differ on at least $d$ (resp. exactly $d$) variables. That is, this problem asks whether there are two satisfying assignments of $\phi$ that overlap on at most $n-d$ (resp. exactly $n-d$) variables. Note that SAT can be reduced to its diverse variant by setting $d$ to $0$. Thus, as SAT is \NPH in general, so is  \textsc{Max/Exact Differ SAT}. So, it is natural to study the diverse variant on those classes of formulas for which SAT is polynomial-time solvable. In particular, we consider affine formulas, $2$-CNF formulas and hitting formulas. We refer to the corresponding restrictions of \textsc{Max/Exact} \textsc{Differ} SAT as \textsc{Max/Exact} \textsc{Differ} \textsc{Affine}-SAT, \textsc{Max/Exact} \textsc{Differ} 2-SAT  and \textsc{Max/Exact Differ Hitting-SAT} respectively. We analyze the classical and parameterized (in parameters $d$ and $n-d$) complexities of these problems.

\textbf{Related work}. This paper is not the first one to study algorithms to determine the maximum number of variables on which two solutions of a given SAT instance can differ. Several exact exponential-time algorithms are known to find a pair of maximally far-apart satisfying assignments. In the mid 2000s, Angelsmark and Thapper \cite{angelsmark2004algorithms} devised an $\mathcal{O}(1.7338^n)$ time algorithm to solve \textsc{Max Hamming Distance 2-SAT}. Their algorithm involved a careful analysis of the micro-structure graph and used a solver for weighted 2-SAT as a sub-routine. Around the same time, Dahl\"off \cite{dahllof2005algorithms} proposed an $\mathcal{O}(1.8348^n)$ time algorithm for \textsc{Max Hamming Distance XSAT}. In the late 2010s, follow-up works of Hoi et. al. \cite{hoi2019measure, hoi2019fast} developed algorithms for the same problem with improved running times, i.e., $\mathcal{O}(1.4983^n)$ for the general case, and $\mathcal{O}(1.328^n)$ for the case when every clause has at most three literals.

\begin{table}[]
\begin{tabular}{l|l|l|l}

                                                     & Classical complexity     & Parameter $d$ & Parameter $n-d$ \\ \hline
Affine formulas         & \NPH, even on $(3,4)$-affine formulas &     \WOH         &     \WOH            \\
& (\cref{Exact Differ Affine-SAT NP-hard}) & (\cref{Exact Differ Affine-SAT W[1]-hard in d}) & (\cref{Exact Differ Affine-SAT W[1]-hard in n-d})

\\
& Polynomial-time on $2$-affine formulas& & \\
& (\cref{Exact Differ Affine-SAT poly-time on 2-affine formulas}) & & \\
\hline
2-CNF formulas           & Polynomial-time on $(2,2)$-CNF formulas &      \WOH  &   ~~~~~~~~~~~~?              \\
&  (\cref{Exact Differ 2-SAT poly-time on (2 2)-CNF formulas})   &  (\cref{Exact Differ 2-SAT W[1]-hard in d})               & 
\\
\hline
Hitting formulas & Polynomial-time & ~~~~~~~~$~_{\mathlarger{-}}$ & ~~~~~~~~~~~$~_{\mathlarger{-}}$\\
&
(\cref{Exact Differ Hitting-SAT poly-time}) &  & \\
\end{tabular}
\caption{Classical and parameterized (in parameters $d$ and $n-d$) complexities of ${\textsc{Exact Differ SAT}}$, when restricted to affine formulas, 2-CNF formulas and hitting formulas.}
\label{table:1} 
\begin{tabular}{l|l|l|l}

                                                     & Classical complexity     & Parameter $d$ & Parameter $n-d$ \\ \hline
Affine formulas           & \NPH, even on $4$-affine formulas &    Single-exponential \FPT           &  \WOH               \\
& (\cref{Max Differ Affine-SAT NP-hard}) &     (\cref{Max Differ Affine-SAT FPT in d})                      & (\cref{Exact Differ Affine-SAT W[1]-hard in n-d}) \\
 & Polynomial-time on $2$-affine formulas & Polynomial kernel &  
\\
& (\cref{Exact Differ Affine-SAT poly-time on 2-affine formulas}) & (\cref{Max Differ Affine-SAT poly-kernel in d})    &
\\
\hline
2-CNF formulas &  Polynomial-time on $(2,2)$-CNF formulas      & \WOH               & ~~~~~~~~~~~~?
\\
&  (\cref{Max Differ 2-SAT poly-time (2 2)-CNF formulas})   &  (\cref{Exact Differ 2-SAT W[1]-hard in d})               &
\\
\hline
Hitting formulas & Polynomial-time & ~~~~~~~~~~~~~~$~_{\mathlarger{-}}$  & ~~~~~~~~~~~$~_{\mathlarger{-}}$ \\
& (\cref{Exact Differ Hitting-SAT poly-time}) & & \\
\end{tabular}
\caption{Classical and parameterized (in parameters $d$ and $n-d$) complexities of \textsc{Max Differ SAT}, when restricted to affine formulas, 2-CNF formulas and hitting formulas.}
\label{table:2}
\end{table}

\textbf{Parameterized complexity}. In the 1990s, Downey and Fellows \cite{downey1992fixed} laid the foundations of parameterized algorithmics. This framework measures the running time of an algorithm as a function of both the input size and a \emph{parameter} $k$, i.e., a suitably chosen attribute of the input. Such a fine-grained analysis helps to cope with the lack of polynomial-time algorithms for \NPH problems by instead looking for an algorithm with running time whose super-polynomial explosion is confined to the parameter $k$ alone. That is, such an algorithm has a running time of the form $f(k)\cdot n^{\mathcal{O}(1)}$, where $f(\cdot)$ is any computable function (could be exponential, or even worse) and $n$ denotes the input size. Such an algorithm is said to be \emph{fixed-parameter tractable} (\FPT) because its running time is polynomially-bounded for every fixed value of the parameter $k$. We refer readers to the textbook `Parameterized Algorithms' by Cygan et. al. \cite{cygan2015parameterized} for an introduction to this field.

\textbf{Our findings}.
We summarize our findings in \cref{table:1} and \cref{table:2}. In \cref{sec: affine formulas}, we show that
\begin{itemize}[noitemsep, nosep]
\item \textsc{Exact Differ Affine-SAT} is \NPH, even on $(3,4)$-affine formulas,
\item \textsc{Max Differ Affine-SAT} is \NPH, even on $4$-affine formulas, 
\item \textsc{Exact/Max Differ Affine-SAT} is polynomial-time solvable on $2$-affine formulas,
\item \textsc{Exact Differ Affine-SAT} is \WOH in the parameter $d$,
\item \textsc{Max Differ Affine-SAT} admits a single-exponential \FPT algorithm in the parameter $d$,
\item \textsc{Max Differ Affine-SAT} admits a polynomial kernel in the parameter $d$, and
\item \textsc{Exact/Max Differ Affine-SAT} is \WOH in the parameter $n-d$.
\end{itemize}
In \cref{sec: 2-CNF formulas}, we show that \textsc{Exact/Max Differ 2-SAT} can be solved in polynomial-time on $(2,2)$-CNF formulas, and \textsc{Exact/Max Differ 2-SAT} is \WOH in the parameter $d$. In \cref{sec: hitting formulas}, we show that \textsc{Exact/Max Differ Hitting-SAT} is polynomial-time solvable. 

\section{Preliminaries}
A Boolean variable can take one of the two truth values: $0$ (False) and $1$ (True). We use $n$ to denote the number of variables in a Boolean formula $\phi$. An \emph{assignment} of $\phi$ is a mapping from the set of all its $n$ variables to $\{0,1\}$. A \emph{satisfying assignment} of $\phi$ is an assignment $\sigma$ such that $\phi$ evaluates to $1$ under $\sigma$, i.e., when every variable $x$ is substituted with its assigned truth value $\sigma(x)$. We say that two assignments $\sigma_1$ 
 and $\sigma_2$ \emph{differ} on a variable $x$ if they assign different truth values to $x$. That is, one of them sets $x$ to $0$, and the other sets $x$ to $1$. Otherwise, we say that $\sigma_1$ and $\sigma_2$ \emph{overlap} on $x$. That is, either both of them set $x$ to $0$, or both of them set $x$ to $1$. 
 
 A \emph{literal} is either a variable $x$ (called a \emph{positive literal}) or its negation $\neg x$ (called a \emph{negative literal}). A \emph{clause} is a disjunction (denoted by $\vee$) of literals. A Boolean formula in \emph{conjunctive normal form}, i.e., a conjunction (denoted by $\wedge$) of clauses, is called a \emph{CNF formula}. A $2$-\emph{CNF formula} is a CNF formula with at most two literals per clause. A $(2,2)$-\emph{CNF formula} is a $2$-CNF formula in which each variable appears in at most two clauses. An \emph{affine formula} is a conjunction of linear equations over the two-element field $\mathbb{F}_2$. We use $\oplus$ to denote the XOR operator, i.e., addition-modulo-2. A $2$-\emph{affine formula} is an affine formula in which each equation has at most two variables. Similarly, a $3$-\emph{affine} (resp. $4$-\emph{affine}) \emph{formula} is an affine formula in which each equation has at most three (resp. four) variables. A $(3,4)$-\emph{affine formula} is a $3$-affine formula in which each variable appears in at most four equations. 

The solution set of a system of linear equations can be obtained in polynomial-time using Gaussian elimination \cite{grcar2011ordinary}. It may have no solution, a unique solution or multiple solutions. When it has multiple solutions, the solution set is described as follows: Some variables are allowed to take any value; we call them \emph{free variables}. The remaining variables take values that are dependent on the values taken by the free variables; we call them \emph{forced variables}. That is, the value taken by any forced variable is a linear combination of the values taken by some free variables. For example, consider the following system of three linear equations over $\mathbb{F}_2$: $x\oplus y\oplus z=1$, $u\oplus y=1, w\oplus z=1$. This system has multiple solutions, and its solution set can be described as $\big\{\big(x,y,z,u,w\big)~|~y\in \mathbb{F}_2, z\in \mathbb{F}_2,  x=y\oplus z\oplus 1,  u=y\oplus 1, 
w=z\oplus 1 \big\}$. Here, $y$ and $z$ are free variables. The remaining variables, i.e., $x, u$ and $w$, are forced variables.

A \emph{hitting formula} is a CNF formula such that for any pair of its clauses, there is some variable that appears as a positive literal in one clause, and as a negative literal in the other clause. That is, no two of its clauses can be simultaneously falsified. Note that the number of unsatisfying assignments of a hitting formula $\phi$ on $n$ variables can be expressed as follows: 
\begin{equation*}
\underset{C:~C \mbox{ \small{is a clause of} } \phi}{\sum}\big|\big\{\sigma~|~\sigma \mbox{ is an assignment of } \phi \mbox{ that falsifies } C\big\}\big| = \underset{C:~C \mbox{ \small{is a clause of} } \phi}{\sum}2^{n-|vars(C)|}
\end{equation*}
Here, we use $vars(C)$ to denote the set of all variables that appear in the clause $C$. 

We use the following as source problems in our reductions:
\begin{itemize}[noitemsep]
\item \textsc{Independent Set}. Given a graph $G$ and a positive integer $k$, decide whether $G$ has an independent set of size $k$. This problem is known to be \NPH on cubic graphs \cite{mohar2001face}, and \WOH in the parameter $k$ \cite{downey2013fundamentals}.
\vspace{0.1 cm}
\item \textsc{Multicolored Clique}. Given a graph $G$ whose vertex set is partitioned into $k$ color-classes, decide whether $G$ has a $k$-sized clique that picks exactly one vertex from each color-class. This problem is known to be \NPH on $r$-regular graphs \cite{cygan2015parameterized}.
\vspace{0.1 cm}
\item \textsc{Exact Even Set}. Given a universe $\mathcal{U}$, a family $\mathcal{F}$ of subsets of $\mathcal{U}$ and a positive integer $k$, decide whether there is a set $X\subseteq \mathcal{U}$ of size exactly $k$ such that $|X\cap S|$ is even for all sets $S$ in the family $\mathcal{F}$. This problem is known to be \WOH in the parameter $k$ \cite{downey1999parametrized}. \vspace{0.1 cm}
\item \textsc{Odd Set} (resp. \textsc{Exact Odd Set}). Given a universe $\mathcal{U}$, a family $\mathcal{F}$ of subsets of $\mathcal{U}$ and a positive integer $k$, decide whether there is a set $X\subseteq \mathcal{U}$ of size at most $k$ (resp. exactly $k$) such that $|X\cap S|$ is odd for all sets $S$ in the family $\mathcal{F}$. Both these problems are known to be \WOH in the parameter $k$ \cite{downey1999parametrized}.
\end{itemize}
We use a polynomial-time algorithm for the following problem as a sub-routine: 
\begin{itemize}[noitemsep]
\item \textsc{Subset Sum problem}. Given a multi-set of integers $\big\{w_1,\ldots, w_p\big\}$ and a target sum $k$, decide whether there exists $X\subseteq [p]$ such that $\sum_{i\in X}w_i = k$. This problem is known to be polynomial-time solvable when the input integers are specified in unary \cite{koiliaris2019faster}.
\end{itemize}
We use the notation $\mathcal{O}^{\star}(\cdot)$ to hide polynomial factors in running time.
\section{Affine formulas}
\label{sec: affine formulas}

In this section, we focus on \textsc{Exact Differ Affine-SAT}, i.e, finding two different solutions to affine formulas. To begin with, we show that finding two solutions that differ on \emph{exactly} $d$ variables is hard even for $(3,4)$-affine formulas: recall that these are instances where every equation has at most three variables and every variable appears in at most four equations.

\begin{theorem}
\label{Exact Differ Affine-SAT NP-hard} \textsc{Exact Differ Affine-SAT} is \NPH, even on $(3,4)$-affine formulas.
\end{theorem}
\begin{proof} We describe a reduction from \textsc{Independent Set on Cubic graphs}. Consider an instance $(G,k)$ of \textsc{Independent Set}, where $G$ is a cubic graph. We construct an affine formula $\phi$ as follows: For every vertex $v\in V(G)$, introduce a variable $x_v$, its $3k$ copies (say $x_v^1,\ldots, x_v^{3k}$), and $3k$ equations: $x_v \oplus x_v^{1} = 0$, $x_{v}^1 \oplus x_v^2 = 0, \ldots, x_v^{3k-1}\oplus x_v^{3k} = 0$. For every edge $e=uv\in E(G)$, introduce variable $y_e$ and equation $x_u \oplus x_v \oplus y_e = 0$. We set $d=k\cdot(3k+4)$.
For every vertex $v\in V(G)$, the variable $x_v$ appears in four equations (i.e., $x_v\oplus x_v^1=0$ and the three equations corresponding to the three edges incident to $v$ in $G$), each of $x_v^1, \ldots, x_v^{3k-1}$ appears in two equations, and $x_v^{3k}$ appears in one equation. For every edge $e\in E(G)$, the variable $y_e$ appears in one equation. So, overall, every variable appears in at most four equations. Also, the equation corresponding to any edge contains three variables, and the remaining equations contain two variables each. Therefore, $\phi$ is a $(3,4)$-affine formula.

Now, we prove that $(G,k)$ is a YES instance of \textsc{Independent Set} if and only if $(\phi, d)$ is a YES instance of \textsc{Exact Differ Affine-SAT}. At a high level, we argue this equivalence as follows: In the forward direction, we show that the two desired satisfying assignments are the all $0$ assignment, and the assignment that i) assigns $1$ to every $x$ variable (and also, its $3k$ copies) that corresponds to a vertex of the independent set, ii) assigns $1$ to every $y$ variable that corresponds to an edge that has one endpoint inside the independent set and the other endpoint outside it, iii) assigns $0$ to every $x$ variable (and also, its $3k$ copies) that corresponds to a vertex outside the independent set, and iv) assigns $0$ to every $y$ variable that corresponds to an edge that has both its endpoints outside the independent set. In the reverse direction, we show that the desired $k$-sized independent set consists of those vertices that correspond to the $x$ variables on which the two assignments differ. We make this argument precise below.

\textbf{Forward direction}. Suppose that $G$ has a $k$-sized independent set, say $S$. Let $\sigma_1$ and $\sigma_2$ be assignments of $\phi$ defined as follows: For every vertex $v\in V(G)\setminus S$, both $\sigma_1$ and $\sigma_2$ set $x_v, x_v^1, \ldots, x_v^{3k}$ to $0$. For every vertex $v\in S$, $\sigma_1$ sets $x_v, x_v^1,\ldots, x_v^{3k}$ to $0$, and $\sigma_2$ sets $x_v, x_v^1, \ldots, x_v^{3k}$ to $1$. For every edge $e\in E(G)$ that has both its endpoints in $V(G)\setminus S$, both $\sigma_1$ and $\sigma_2$ set $y_e$ to $0$. For every edge $e\in E(G)$ that has one endpoint in $S$ and the other endpoint in $V(G)\setminus S$, $\sigma_1$ sets $y_e$ to $0$, and $\sigma_2$ sets $y_e$ to $1$. 

As $\sigma_1$ sets all variables to $0$, it is clear that it satisfies $\phi$. Now, we show that $\sigma_2$ satisfies $\phi$. Consider any edge $e=uv\in E(G)$ and its corresponding equation $x_u \oplus x_v \oplus y_e = 0$. If both endpoints of $e$ belong to $V(G)\setminus S$, then $\sigma_2$ sets $x_u$, $x_v$ and $y_e$ to $0$. Also, if $e$ has one endpoint (say $u$) in $S$, and the other endpoint in $V(G)\setminus S$, then $\sigma_2$ sets $x_u$ to $1$, $x_v$ to $0$ and $y_e$ to $1$. Therefore, in both cases, $x_u \oplus x_v \oplus y_e$ takes the truth value $0$ under $\sigma_2$. Also, for any vertex $v\in V(G)$, since $\sigma_2$ gives the same truth value to $x_v, x_v^1, \ldots, x_v^{3k}$ (i.e., all $1$ if $v\in S$, and all $0$ if $v\in V(G)\setminus S$), it also satisfies the equations $x_v\oplus x_v^1=0, x_v^1\oplus x_v^2=0,\ldots, x_v^{3k-1}\oplus x_v^{3k}=0$. Thus, $\sigma_2$ is a satisfying assignment of $\phi$.

As $G$ is a cubic graph, every vertex in $S$ is incident to three edges in $G$. Also, as $S$ is an independent set, none of these edges has both endpoints in $S$. Therefore, there are $3\cdot |S|$ edges that have one endpoint in $S$ and the other endpoint in $V(G)\setminus S$. Note that $\sigma_1$ and $\sigma_2$ differ on the $y$ variables that correspond to these $3\cdot |S|$ edges. Also, they differ on $|S|$ many $x$ variables, and their $3k\cdot |S|$ copies. Therefore, overall, they differ on $(3k+1)\cdot |S| + 3\cdot |S| = k\cdot (3k+4)$ variables. Hence, $(\phi,d)$ is a YES instance of \textsc{Exact Differ Affine SAT}.

\textbf{Reverse direction}. Suppose that $(\phi,d)$ is a YES instance of \textsc{Exact Differ Affine-SAT}. That is, there exist satisfying assignments $\sigma_1$ and $\sigma_2$ of $\phi$ that differ on $k\cdot (3k+4)$ variables. Let $S:=\big\{v\in V(G)~|~\sigma_1 \mbox{ and } \sigma_2 \mbox{ differ on } x_v\big\}$. We show that $S$ is a $k$-sized independent set of $G$. Let $e(S, \bar{S})$ denote the number of edges in $G$ that have one endpoint in $S$ and the other endpoint in $V(G)\setminus S$. Now, let us express the number of variables on which $\sigma_1$ and $\sigma_2$ differ in terms of $|S|$ and $e(S,\bar{S})$. 

Consider any edge $e=uv\in E(G)$. First, suppose that $e$ has both its endpoints in $S$. Then, as $\sigma_1$ and $\sigma_2$ differ on both $x_u$ and $x_v$, the expression $x_u \oplus x_v$ takes the same truth value under $\sigma_1$ and $\sigma_2$. So, as both of them satisfy the equation $x_u\oplus x_v\oplus y_e=0$, it follows that $\sigma_1$ and $\sigma_2$ must overlap on $y_e$. Next, suppose that $e$ has both its endpoints in $V(G)\setminus S$. Then, as $\sigma_1$ and $\sigma_2$ overlap on both $x_u$ and $x_v$, the expression $x_u\oplus x_v$ takes the same truth value under $\sigma_1$ and $\sigma_2$. So, again, $\sigma_1$ and $\sigma_2$ must overlap on $y_e$. Next, suppose that $e$ has one endpoint (say $u$) in $S$ and the other endpoint in $V(G)\setminus S$. Then, as $\sigma_1$ and $\sigma_2$ differ on $x_u$ and overlap on $x_v$, the expression $x_u\oplus x_v$ takes different truth values under $\sigma_1$ and $\sigma_2$. So, as both $\sigma_1$ and $\sigma_2$ satisfy the equation $x_u\oplus x_v\oplus y_e=0$, it follows that $\sigma_1$ and $\sigma_2$ must differ on $y_e$. So, overall, $\sigma_1$ and $\sigma_2$ differ on $e(S,\bar{S})$ many $y$ variables.

For any vertex $v\in V(G)$, since any satisfying assignment satisfies the equations $x_v \oplus x_v^1 = 0, x_v^1 \oplus x_v^2 = 0, \ldots, x_v^{3k-1}\oplus x_v^{3k}=0$, it must assign the same truth value to $x_v, x_v^1, \ldots, x_v^{3k}$. So, for any $v\in S$, as $\sigma_1$ and $\sigma_2$ differ on $x_v$, they also differ on $x_v^1, \ldots, x_v^{3k}$. Similarly, for any $v\in V(G)\setminus S$, as $\sigma_1$ and $\sigma_2$ overlap on $x_v$, they also overlap on $x_v^1,\ldots x_v^{3k}$. So, overall, $\sigma_1$ and $\sigma_2$ differ on $|S|$ many $x$ variables and their $3k\cdot |S|$ copies. Now, summing up the numbers of $y$ variables and $x$ variables (and their copies) on which $\sigma_1$ and $\sigma_2$ differ, we get 
\begin{equation}
\label{eq1}
e(S,\bar{S})+ (3k+1)\cdot |S| = k\cdot (3k+4)
\end{equation}
Let $e(S,S)$ denote the number of edges in $G$ that have both endpoints in $S$. Note that $$\sum_{v\in S}\mbox{degree}_G(v) = 2\cdot e(S,S)+ e(S,\bar{S})$$
Also, as $G$ is a cubic graph, we know that $\mbox{degree}_G(v)=3$ for all $v\in S$. Therefore, we get $e(S,\bar{S}) = 3\cdot |S| - 2\cdot e(S,S)$. Putting this expression for $e(S,\bar{S})$ in \cref{eq1}, we have
\begin{equation}
\label{eq2}
(3k+4)\cdot \big(|S|-k\big) = 2\cdot e(S,S)
\end{equation}
If $|S|\geq k+1$, then LHS of \cref{eq1} becomes $\geq (3k+1)\cdot (k+1) = k\cdot (3k+4) + 1$, which is greater than its RHS. So, we must have $|S|\leq k$. Also, as RHS of \cref{eq2} is non-negative, so must be its LHS. This gives us $|S|\geq k$.  Therefore, it follows that $|S|=k$. Putting $|S|=k$ in \cref{eq2}, we also get $e(S,S)=0$. That is, $S$ is an independent set of $G$. Hence, $(G,k)$ is a YES instance of \textsc{Independent Set}.

This proves \cref{Exact Differ Affine-SAT NP-hard}.
\end{proof}

We now turn to \textsc{Max Differ Affine-SAT}, i.e, finding two solutions that differ on \emph{at least} $d$ variables. We show that this is hard for affine formulas of bounded arity.

\begin{theorem}
\label{Max Differ Affine-SAT NP-hard} \textsc{Max Differ Affine-SAT} is \NPH, even on $4$-affine formulas.
\end{theorem}

\begin{proof}
We describe a reduction from \textsc{Multicolored Clique on Regular graphs}. Consider an instance $(G,k)$ of \textsc{Multicolored Clique}, where $G$ is a $r$-regular graph. We assume that each color-class of $G$ has size $N:= 2\cdot 3^q$. It can be argued that a suitably-sized $r$-regular graph exists whose addition to the color-class makes this assumption hold true. We construct an affine formula $\phi$ as follows: For every vertex $v\in V(G)$, introduce a variable $x_v$ and its $\ell$ copies $\big($say $x_v^1, x_v^2, \ldots x_v^{\ell}\big)$, where $\ell:= k\cdot (r-k+1)+k\cdot q$. We force these copies to take the same truth value as $x_v$ via the equations $x_v \oplus x_v^1 = 0, x_v^1\oplus x_v^2 =0, \ldots, x_v^{\ell-1}\oplus x_v^{\ell}=0$. For every edge $e=uv\in E(G)$, we add variables $y_e$ and $z_e$, and also the equation  $x_u\oplus x_v\oplus y_e\oplus z_e = 1$. 

For any $1\leq i\leq k$, consider the $i^{th}$ color-class, say $V_i = \{v_i^1, v_i^2, \ldots, v_i^{N}\}$. First, we add $N/3$ many 
\emph{Stage 1 dummy variables} $\big($say $d_{i,1}^1$, $d_{i,1}^2$,$\ldots$,$d_{i,1}^{N/3}\big)$, group the $x$ variables corresponding to the vertices of $V_i$ into $N/3$ triplets, and add $N/3$ equations that equate the xor of a triplet's variables and a dummy variable to $0$. More precisely, we add the following $N/3$ equations: 
\begin{equation*}
\big(x_{v_i^1} \oplus x_{v_i^2} \oplus x_{v_i^3}\big) \oplus d_{i,1}^1 = 0,~ \big(x_{v_i^4} \oplus x_{v_i^5} \oplus x_{v_i^6}\big) \oplus d_{i,1}^2 = 0, ~\ldots,~ \big(x_{v_i^{N-2}} \oplus x_{v_i^{N-1}} \oplus x_{v_i^{N}}\big) \oplus d_{i,1}^{N/3} = 0
\end{equation*}
Next, we repeat the same process as follows: We introduce $N/3^2$ many \emph{Stage 2 dummy variables} $\big($say $d_{i,2}^1$, $d_{i,2}^2$, $\ldots$, $d_{i,2}^{N/3^2}\big)$, group the $N/3$ many Stage 1 dummy variables into $N/3^2$ triplets, and add $N/3^2$ equations that equate the xor of a triplet's Stage 1 dummy variables and a Stage 2 dummy variable to $0$. More precisely, we add the following $N/3^2$ equations:
\begin{equation*}
\resizebox{1\hsize}{!}{
$\big(d_{i,1}^1 \oplus d_{i,1}^2 \oplus d_{i,1}^3\big) \oplus d_{i,2}^1 = 0,~ \big(d_{i,1}^4 \oplus d_{i,1}^5 \oplus d_{i,1}^6\big) \oplus d_{i,2}^2 = 0, ~\ldots,~ \big(d_{i,1}^{N/3-2} \oplus d_{i,1}^{N/3-1} \oplus d_{i,1}^{N/3}\big) \oplus d_{i,2}^{N/3^2} = 0$}
\end{equation*}
Repeating the same procedure, we add the following $N/3^3$, $N/3^4$, $\ldots$, $N/3^q=2$ equations:
\begin{equation*}
\resizebox{1\hsize}{!}{
$\big(d_{i,2}^1 \oplus d_{i,2}^2 \oplus d_{i,2}^3\big) \oplus d_{i,3}^1 = 0,~ \big(d_{i,2}^4 \oplus d_{i,2}^5 \oplus d_{i,2}^6\big) \oplus d_{i,3}^2 = 0, ~\ldots,~ \big(d_{i,2}^{N/3^2-2} \oplus d_{i,2}^{N/3^2-1} \oplus d_{i,2}^{N/3^2}\big) \oplus d_{i,3}^{N/3^3} = 0$}
\end{equation*}
\begin{equation*}
\resizebox{1\hsize}{!}{
$\big(d_{i,3}^1 \oplus d_{i,3}^2 \oplus d_{i,3}^3\big) \oplus d_{i,4}^1 = 0,~ \big(d_{i,3}^4 \oplus d_{i,3}^5 \oplus d_{i,3}^6\big) \oplus d_{i,4}^2 = 0, ~\ldots,~ \big(d_{i,3}^{N/3^3-2} \oplus d_{i,3}^{N/3^3-1} \oplus d_{i,3}^{N/3^3}\big) \oplus d_{i,4}^{N/3^4} = 0$}
\end{equation*}
$$\vdots$$
\begin{equation*}
\big(d_{i,q-1}^1 \oplus d_{i,q-1}^2 \oplus d_{i,q-1}^3\big) \oplus d_{i,q}^1 = 0,~ \big(d_{i,q-1}^4 \oplus d_{i,q-1}^5 \oplus d_{i,q-1}^6\big) \oplus d_{i,q}^2 = 0
\end{equation*}
Next, we add $B+1$ \emph{auxiliary variables} $\big($say $D_i^1, \ldots, D_i^{B+1}\big)$ and the following equations: $$\big(d_{i,q}^1\oplus d_{i,q}^2\big) \oplus D_i^1 = 0, ~\big(d_{i,q}^1\oplus d_{i,q}^2\big) \oplus D_i^2 = 0,~\ldots,~
\big(d_{i,q}^1\oplus d_{i,q}^2\big) \oplus D_i^{B+1} = 0,$$
where $B:= k\cdot (\ell+1)+k\cdot (r-k+1)+k\cdot q$ is the budget that we set on the total number of overlaps. That is, we set $d = n-B$, where $n$ denotes the number of variables in $\phi$. Now, we prove that $(G,k)$ is a YES instance of \textsc{Multicolored Clique} if and only if $(\phi,d)$ is a YES instance of \textsc{Max Differ Affine-SAT}. We present a proof sketch of this equivalence below.

\textbf{Forward direction}. In the first assignment, we set i) all $x$ and $y$ variables to $0$, ii) all $z$ variables to $1$, and iii) all dummy and auxiliary variables to $0$. In the second assignment, we assign i) $0$ to the $k$ many $x$ variables that correspond to the multi-colored clique's vertices, ii) $1$ to the remaining $x$ variables, iii) $0$ to all $z$ variables, iv) $0$ to the $k\cdot (r-k+1)$ many $y$ variables that correspond to those edges that have one endpoint inside the multi-colored clique and the other endpoint outside it, v) $1$ to the remaining $y$ variables, and vi) $1$ to all auxiliary variables. Also, in the second assignment, for each $1\leq i\leq k$, we assign i) $0$ to that Stage $1$ dummy variable which was grouped with the $x$ variable corresponding to the multi-colored clique's vertex from the $i^{th}$ color-class, $0$ to that Stage $2$ dummy variable which was grouped with this Stage $1$ dummy variable, $0$ to that Stage $3$ dummy variable which was grouped with this Stage $2$ dummy variable, and so on $\ldots$, and ii) $1$ to the remaining dummy variables. It can be verified that these two assignments satisfy $\phi$, and they overlap on $B$ many variables.

\textbf{Reverse direction}. First, we show that each of the $k$ color-classes has at least one vertex on whose corresponding $x$ variable the two assignments overlap. Consider any $1\leq i\leq k$. Since the $B+1$ auxiliary variables are forced to take the same truth value and there are only at most $B$ overlaps, the two assignments must differ on them. This forces the two assignments to overlap on one of the two Stage $q$ dummy variables. Further, this forces at least one overlap amongst the three Stage $q-1$ dummy variables that were grouped with this Stage $q$ dummy variable. This effect propagates to lower-indexed stages, and eventually forces at least one overlap amongst the $x$ variables corresponding to the vertices of the $i^{th}$ color-class.

Next, we show that each of the $k$ color-classes has at most one vertex on whose corresponding $x$ variable the two assignments overlap. Suppose not. Then, there are at least two overlaps amongst the $x$ variables corresponding to the vertices of some color class. Also, based on the previous paragraph, we know that there is at least one overlap amongst the $x$ variables corresponding to the vertices of each of the remaining $k-1$ color classes. Therefore, overall, there are at least $k+1$ many overlaps amongst the $x$ variables. So, the contribution of these $x$ variables and their copies to the total number of overlaps becomes $\geq (k+1)\cdot (\ell+1)=B+1$. However, this exceeds the budget $B$ on the number of overlaps, which is a contradiction.  

Based on the previous two paragraphs, we know that for each $1\leq i\leq k$, there is exactly one overlap amongst the $x$ variables corresponding to the vertices of the $i^{th}$ color class.  Finally, we show that the set, say $S$, formed by these $k$ vertices is the desired multi-colored clique. Suppose not. Then, there are $>k\cdot (r-k+1)$ edges that have one endpoint in $S$ and the other endpoint outside $S$. Also, for each such edge, the two assignments must overlap on one of its corresponding $y$ and $z$ variables. Therefore, we have $>k\cdot (r-k+1)$ overlaps on the $y$ and $z$ variables. Also, $k\cdot q$ overlaps are forced on the dummy variables via the equations added in the grouping procedure. Thus, overall, the total number of overlaps exceeds $B$, which is a contradiction. 

This concludes a proof sketch of \cref{Max Differ Affine-SAT NP-hard}.
\end{proof}

If, on the other hand, all equations in the formula have at most two variables, then both problems turn out to be tractable. We describe this algorithm next.

\begin{theorem}
\label{Exact Differ Affine-SAT poly-time on 2-affine formulas}
\textsc{Exact/Max Differ Affine-SAT} is polynomial-time solvable on $2$-affine formulas.
\end{theorem}
\begin{proof}
Consider an instance $(\phi,d)$ of \textsc{Exact Differ Affine-SAT}, where $\phi$ is a $2$-affine formula. First, we construct a graph $G_{0}$ as follows: Introduce a vertex for every variable of $\phi$. For every equation of the form $x\oplus y=0$ in $\phi$, add the edge $xy$. We compute the connected components of $G_{0}$. Observe that for each component $C$ of $G_{0}$, the equations of $\phi$ corresponding to the edges of $C$ are simultaneously satisfied if and only if all variables of $C$ take the same truth value. So, any pair of satisfying assignments of $\phi$ either overlap on all variables in $C$, or differ on all variables in $C$. Thus, we replace all variables in $C$ by a single variable, and set its weight to be the size of $C$. More precisely, i) we remove all but one variable (say $z$) of $C$ from the variable-set of $\phi$, ii) we remove all those equations from $\phi$ that correspond to the edges of $C$, iii) for every variable $v\in C\setminus \{z\}$, we replace the remaining appearances of $v$ in $\phi$ (i.e., in equations of the form $v\oplus \underline{~~} =1$) with $z$, and iv) we set the weight of $z$ to be the number of variables in $C$. Let $\phi'$ denote the variable-weighted affine formula so obtained. Then, our goal is to decide whether $\phi'$ has a pair of satisfying assignments such that the weights of the variables at which they differ add up to exactly $d$. 

Note that all equations in $\phi'$ are of the form $x\oplus y=1$. Next, we construct a vertex-weighted graph $G_{1}$ as follows: Introduce a vertex for every variable of $\phi'$, and assign it the same weight as that of its corresponding variable. For every equation $x\oplus y=1$ of $\phi'$, add the edge $xy$. We compute the connected components of $G_{1}$. Then, we run a bipartiteness-testing algorithm on each component of $G_{1}$. Suppose that there is a component $C$ of $G_{1}$ that is not bipartite. Then, there is an odd-length cycle in $C$, say with vertices $x_1,x_2, \ldots, x_{2\ell}, x_{2\ell+1}$ (in that order). Note that the edges of this cycle correspond to the equations $x_1 \oplus x_2=1$, $x_2\oplus x_3=1,\ldots, x_{2\ell}\oplus x_{2\ell+1}=1$, $x_{2\ell+1}\oplus x_1=1$ in $\phi'$. Adding (modulo $2$) these $2\ell+1$ equations, we get LHS $= (2\cdot x_1 + 2\cdot x_2 + \ldots + 2\cdot x_{2\ell+1})~mod~2 = 0$, and RHS $= (2\ell+1)~mod~2 = 1$. So, these $2\ell+1$ equations of $\phi'$ cannot be simultaneously satisfied. Thus, we return NO. Now, assume that all components of $G_{1}$ are bipartite. See \cref{fig: graph G1 for 2-affine formulas} for an example.

\begin{figure}
\begin{center}
\includegraphics[scale=0.7]{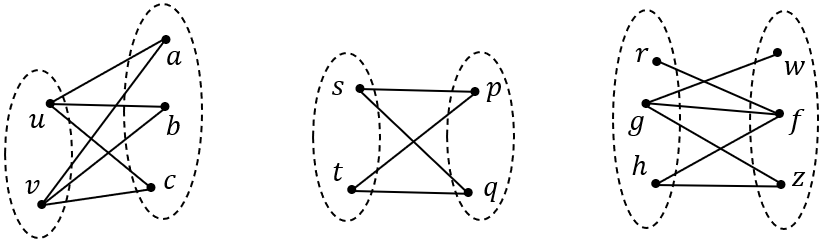}
\end{center}
\caption{This figure shows the bipartite components of the graph $G_1$ constructed in the proof of \cref{Exact Differ Affine-SAT poly-time on 2-affine formulas}, when the $2$-affine formula $\phi'$ consists of the following equations: $u\oplus a=1$, $u\oplus b=1$, $u\oplus c=1$, $v\oplus a=1$, $v\oplus b=1$, $v\oplus c=1$, $s\oplus p=1$, $s\oplus q=1$, $t\oplus p=1$, $t\oplus q=1$, $r\oplus f=1$, $g\oplus w=1$, $g\oplus f=1$, $g\oplus z=1$, $h\oplus f=1$, $h\oplus z=1$.} 
\label{fig: graph G1 for 2-affine formulas}
\end{figure}

Let $C_1,\ldots, C_{p}$ denote the connected components of $G_1$. Consider any $1\leq i\leq p$. Let $A$ and $B$ denote the parts of the bipartite component $C_i$. Observe that the equations of $\phi'$ corresponding to the edges of $C_i$ are simultaneously satisfied if and only if either i) all variables in $A$ are set to $1$, and all variables in $B$ are set to $0$, or ii) all variables in $A$ are set to $0$, and all variables in $B$ are set to $1$. So, any pair of satisfying assignments of $\phi'$ either overlap on all variables in $C_i$, or differ on all variables in $C_i$. Thus, our problem amounts to deciding whether there is a subset of components of $G_1$ whose collective weight is exactly $d$. That is, our goal is to decide whether there exists $X\subseteq [p]$ such that $\sum_{i\in X} weight(C_i) = d$, where $weight(C_i)$ denotes the sum of the weights of the variables in $C_i$. To do so, we use the algorithm for \textsc{Subset Sum problem} with $\big\{weight(C_1),\ldots, weight(C_p)\big\}$ as the multi-set of integers and $d$ as the target sum. 

The algorithm described here works almost as it is for \textsc{Max Differ Affine-SAT} too. In the last step, instead of reducing to \textsc{Subset Sum problem}, we simply check whether the collective weight of all components of $G_1$ is at least $d$. That is, if $\sum_{i=1}^p weight(C_i)\geq d$, we return YES; otherwise, we return NO. Thus, both \textsc{Exact Differ Affine-SAT} and \textsc{Max Differ Affine-SAT} are polynomial-time solvable on $2$-affine formulas. This proves \cref{Exact Differ Affine-SAT poly-time on 2-affine formulas}.
\end{proof}


We now turn to the parameterized complexity of \textsc{Exact Differ Affine-SAT} and \textsc{Max Differ Affine-SAT} when parameterized by the number of variables that differ in the two solutions. It turns out that the exact version of the problem is \WOH{}, while the maximization question is \FPT{}. We first show the hardness of \textsc{Exact Differ Affine-SAT} by a reduction from \textsc{Exact Even Set}.

\begin{theorem}
\label{Exact Differ Affine-SAT W[1]-hard in d} \textsc{Exact Differ Affine-SAT} is \WOH in the parameter $d$.
\end{theorem}
\begin{proof}
We describe a reduction from \textsc{Exact Even Set}. Consider an instance $(\mathcal{U},\mathcal{F},k)$ of \textsc{Exact Even Set}. We construct an affine formula $\phi$ as follows: For every element $u$ in the universe $\mathcal{U}$, introduce a variable $x_u$. For every set $S$ in the family $\mathcal{F}$, introduce the equation $\underset{u\in S}{\Oplus}x_u = 0$. We set $d=k$. We prove that $(\mathcal{U},\mathcal{F},k)$ is a YES instance of \textsc{Exact Even Set} if and only if $(\phi, d)$ is a YES instance of \textsc{Exact Differ Affine-SAT}. At a high level, we argue this equivalence as follows: In the forward direction, we show that the two desired satisfying assignments are i) the all $0$ assignment, and ii) the assignment that assigns $1$ to the variables that correspond to the elements of the given even set, and assigns $0$ to the remaining variables. In the reverse direction, we show that the desired even set consists of those elements of the universe that correspond to the variables on which the two given satisfying assignments differ. We make this argument precise below.

\textbf{Forward direction}. Suppose that $(\mathcal{U},\mathcal{F},k)$ is a YES instance of \textsc{Exact Even Set}. That is, there is a set $X\subseteq \mathcal{U}$ of size exactly $k$ such that $|X\cap S|$ is even for all sets $S$ in the family $\mathcal{F}$. Let $\sigma_1$ and $\sigma_2$ be assignments of $\phi$ defined as follows: For every $u\in X$, $\sigma_1$ sets $x_u$ to $0$, and $\sigma_2$ sets $x_u$ to $1$. For every $u\in \mathcal{U}\setminus X$, both $\sigma_1$ and $\sigma_2$ set $x_u$ to $0$. Note that $\sigma_1$ and $\sigma_2$ differ on exactly $|X|=k$ variables. Consider any set $S$ in the family $\mathcal{F}$. The  equation corresponding to $S$ in the formula $\phi$ is $\underset{u\in S}{\Oplus}x_u=0$. All variables in the left-hand side are set to $0$ by $\sigma_1$. Also, the number of variables in the left-hand side that are set to $1$ by $\sigma_2$ is $|X\cap S|$, which is an even number. Therefore, the left-hand side evaluates to $0$ under both $\sigma_1$ and $\sigma_2$. So, $\sigma_1$ and $\sigma_2$ are satisfying assignments of $\phi$. Hence, $(\phi,k)$ is a YES instance of \textsc{Exact Differ Affine-SAT}. 

\textbf{Reverse direction}. Suppose that $(\phi,k)$ is a YES instance of \textsc{Exact Differ Affine-SAT}. That is, there are satisfying assignments $\sigma_1$ and $\sigma_2$ of $\phi$ that differ on exactly $k$ variables. Let $X$ denote the $k$-sized set $\big\{u\in \mathcal{U}~|~\sigma_1 \mbox{ and } \sigma_2 \mbox{ differ on } x_u\big\}$. Consider any set $S$ in the family $\mathcal{F}$. The equation corresponding to $S$ in the formula $\phi$ is $\underset{u\in S}{\Oplus}x_u=0$. We split the left-hand side into two parts to express this equation as $\underbracket[0.4pt]{\underset{u\in S\setminus X}{\Oplus}x_u}_{\clap{\scriptsize{A}}}\oplus \underbracket[0.4pt]{\underset{u\in X\cap S}{\Oplus}x_u}_{\clap{\scriptsize{B}}}=0$. Note that $\sigma_1$ and $\sigma_2$ overlap on all variables in the first part, i.e., $A$. So, $A$ evaluates to the same truth value under both assignments. Thus, as both $\sigma_1$ and $\sigma_2$ satisfy this equation, they must assign the same truth value to the second part, i.e., $B$, as well. Also, $\sigma_1$ and $\sigma_2$ differ on all variables in $B$. So, for its truth value to be same under both assignments, $B$ must have an even number of variables. That is, $|X\cap S|$ must be even. Hence, $(\mathcal{U}, \mathcal{F}, k)$ is a YES instance of \textsc{Exact Even Set}. 

This proves \cref{Exact Differ Affine-SAT W[1]-hard in d}.
\end{proof}

We now turn to the FPT algorithm for \textsc{Max Differ Affine-SAT}, which is based on obtaining solutions using Gaussian elimination and working with the free variables: if the set of free variables $F$ is ``large'', we can simply set them differently and force the dependent variables, and guarantee ourselves a distinction on at least $|F|$ variables. Note that this is the step that would not work as-is for the exact version of the problem. If the number of free variables is bounded, we can proceed by guessing the subset of free variables on which the two assignments differ. We make these ideas precise below. 

\begin{theorem}
\label{Max Differ Affine-SAT FPT in d} \textsc{Max Differ Affine-SAT} admits an algorithm with running time $\mathcal{O}^{\star}(2^d)$.
\end{theorem}
\begin{proof}
Consider an instance $(\phi,d)$ of \textsc{Max Differ Affine-SAT}. We use Gaussian elimination to find the solution set of $\phi$ in polynomial-time. If $\phi$ has no solution, we return NO. If $\phi$ has a unique solution and $d=0$, we return YES. If $\phi$ has a unique solution and $d\geq 1$, we return NO. Now, assume that $\phi$ has multiple solutions. 

Let $F$ denote the set of all free variables. Suppose that $|F|\geq d$. Let $\sigma_1$ denote the solution of $\phi$ obtained by setting all free variables to $0$, and then setting the forced variables to take values as per their dependence on the free variables. Similarly, let $\sigma_2$ denote the solution of $\phi$ obtained by setting all free variables to $1$, and then setting the forced variables to take values as per their dependence on the free variables. Note that $\sigma_1$ and $\sigma_2$ differ on all free variables (and possibly some forced variables too). So, overall, they differ on at least $|F|\geq d$ variables. Thus, we return YES. Now, assume that $|F|\leq d-1$. We guess the subset $D\subseteq F$ of free variables on which two desired solutions (say $\sigma_1$ and $\sigma_2$) differ. Note that there are $2^{|F|}\leq 2^{d-1}$ such guesses.

First, consider any forced variable $x$ that depends on an odd number of free variables from $D$. That is, the expression for its value is the XOR of an odd number of free variables from $D$ (possibly along with the constant $1$ and/or some free variables from $F\setminus D$). Then, note that this expression takes different truth values under $\sigma_1$ and $\sigma_2$. That is, $\sigma_1$ and $\sigma_2$ differ on $x$. 

Next, consider any forced variable $x$ that depends on an even number of free variables from $D$. That is, the expression for its value is the XOR of an even number of free variables from $D$ (possibly along with the constant $1$ and/or some free variables from $F\setminus D$). Then, note that this expression takes the same truth value under $\sigma_1$ and $\sigma_2$. That is, $\sigma_1$ and $\sigma_2$ overlap on $x$. 

Thus, overall, these two solutions differ on i) all free variables from $D$, and ii) all those forced variables that depend upon an odd number of free variables from $D$. If the total count of such variables is $\geq d$ for some guess $D$, we return YES. Otherwise, we return NO. This proves \cref{Max Differ Affine-SAT FPT in d}.
\end{proof}
Now, we show that~\textsc{Max Differ Affine-SAT} has a polynomial kernel in the parameter $d$.
\begin{theorem}
\label{Max Differ Affine-SAT poly-kernel in d} \textsc{Max Differ Affine-SAT} admits a kernel with $\mathcal{O}(d^2)$ variables and $\mathcal{O}(d^2)$ equations. \end{theorem}
\begin{proof}
Consider an instance $(\phi,d)$ of \textsc{Max Differ Affine-SAT}. We use Gaussian elimination to find the solution set of $\phi$ in polynomial-time. Then, as in the proof of \cref{Max Differ Affine-SAT FPT in d}, i) we return NO if $\phi$ has no solution, or if $\phi$ has a unique solution and $d\geq 1$, ii) we return YES if $\phi$ has a unique solution and $d=0$, or if $\phi$ has multiple solutions with at least $d$ free variables. Now, assume that $\phi$ has multiple solutions with at most $d-1$ free variables.

Note that the system of linear equations formed by the expressions for the values of forced variables is an affine formula (say $\phi'$) that is equivalent to $\phi$. That is, $\phi'$ and $\phi$ have the same solution sets. So, we work with the instance $(\phi',d)$ in the remaining proof. 

Suppose that there is a free variable, say $x$, such that at least $d-1$ forced variables depend on $x$. That is, there are at least $d-1$ forced variables such that the expressions for their values are the XOR of $x$ (possibly along with the constant $1$ and/or some other free variables). Let $\sigma_1$ denote the solution of $\phi'$ obtained by setting all free variables to $0$, and then setting the forced variables to take values as per their dependence on the free variables. Let $\sigma_2$ denote the solution of $\phi'$ obtained by setting $x$ to $1$ and the remaining free variables to $0$, and then setting the forced variables to take values as per their dependence on the free variables. Note that $\sigma_1$ and $\sigma_2$ differ on $x$, and also on each of the $\geq d-1$ forced variables that depend on $x$. So, overall, $\sigma_1$ and $\sigma_2$ differ on at least $d$ variables. Thus, we return YES. 

Now, assume that for every free variable $x$, there are at most $d-2$ forced variables that depend on $x$. So, as there are at most $d-1$ free variables, it follows that there are at most $(d-1)\cdot (d-2)$ forced variables that depend on at least one free variable. The remaining forced variables are the ones that do not depend on any free variable. That is, any such forced variable $y$ is set to a constant (i.e., $0$ or $1$) as per the expression for its value. We remove the variable $y$ and its corresponding equation (i.e., $y=0$ or $y=1$) from $\phi'$, and we leave $d$ unchanged. This is safe because $y$ takes the same truth value under all solutions of $\phi'$.

Note that the affine formula so obtained has at most $d-1$ free variables and at most $(d-1)\cdot (d-2)$ forced variables. So, overall, it has at most $(d-1)^2$ variables. Also, it has at most $(d-1)\cdot (d-2)$ equations. This proves \cref{Max Differ Affine-SAT poly-kernel in d}.
\end{proof}


We finally turn to the ``dual'' parameter, $n-d$: the number of variables on which the two assignments sought \emph{overlap}. We show that both the exact and maximization variants for affine formulas are \WOH{} in this parameter by reductions from \textsc{Exact Odd Set} and \textsc{Odd Set}, respectively.

\begin{theorem}
\label{Exact Differ Affine-SAT W[1]-hard in n-d} \textsc{Exact/Max Differ Affine-SAT} is \WOH in the parameter $n-d$.
\end{theorem}
\begin{proof}
We describe a reduction from \textsc{Exact Odd Set}. Consider an instance $(\mathcal{U}, \mathcal{F}, k)$ of \textsc{Exact Odd Set}. We construct an affine formula $\phi$ as follows:  For every element $u$ in the universe $\mathcal{U}$, introduce a variable $x_u$. For every odd-sized set $S$ in the family $\mathcal{F}$, introduce the equation $\underset{u\in S}{\Oplus}x_u=1$. For every even-sized set $S$ in the family $\mathcal{F}$, introduce $k+1$ variables $y_S, z_S^{1}, z_{S}^2,\ldots, z_{S}^k$, and the equations $y_S\oplus z_S^{1}=0, y_S\oplus z_S^{2}=0, \ldots, {y_S\oplus z_S^{k}=0}$ and ${\underset{u\in S}{\Oplus}x_u \oplus y_S=0}$. The number of variables in $\phi$ is $n= |\mathcal{U}|+{(k+1)}\cdot$${\big|\big\{S\in \mathcal{F}~\big|~|S| \mbox{ is even}\big\}\big|}$. We set $d=n-k$. 

We prove that $(\mathcal{U}, \mathcal{F}, k)$ is  a YES instance of \textsc{Exact Odd Set} if and only if $(\phi, d)$ is a YES instance of \textsc{Exact Differ Affine}-SAT. At a high level, we argue this equivalence as follows: In the forward direction, we show that the two desired satisfying assignments are i) the assignment that sets all $y$ and $z$ variables to $0$ and all $x$ variables to $1$, and ii) the assignment that sets all $y$ and $z$ variables to $1$, assigns $1$ to all those $x$ variables that correspond to the elements of the given odd set, and assigns $0$ to the remaining $x$ variables. In the reverse direction, we show that the two assignments must differ on all $y$ and $z$ variables (and so, all $k$ overlaps are restricted to occur at $x$ variables), and the desired odd set consists of those elements of the universe that correspond to the $x$ variables on which the two assignments overlap. We make this argument precise below.

\textbf{Forward direction}. Suppose that $(\mathcal{U}, \mathcal{F}, k)$ is a YES instance of \textsc{Exact Differ Affine}-SAT. That is, there is a set $X\subseteq \mathcal{U}$ of size exactly $k$ such that $|X\cap S|$ is odd for all sets $S$ in the family $\mathcal{F}$. Let $\sigma_1$ and $\sigma_2$ be assignments of $\phi$ defined as follows: For every even-sized set $S$ in the family $\mathcal{F}$, $\sigma_1$ sets $y_S, z_S^1, z_S^2,\ldots, z_S^k$ to $0$, and $\sigma_2$ sets $y_S, z_S^1, z_S^2,\ldots, z_S^k$ to $1$. For every $u\in X$, both $\sigma_1$ and $\sigma_2$ set $x_u$ to $1$. For every $u\in \mathcal{U}\setminus X$, $\sigma_1$ sets $x_u$ to 1, and $\sigma_2$ sets $x_u$ to $0$. Note that $\sigma_1$ and $\sigma_2$ overlap on exactly $|X|=k$ variables (and so, they differ on exactly $n-k$ variables). Now, we show that $\sigma_1$ and $\sigma_2$ are satisfying assignments of $\phi$. 

First, we argue that $\sigma_1$ and $\sigma_2$ satisfy the equations of $\phi$ that were added corresponding to odd-sized sets of the family $\mathcal{F}$. Consider any odd-sized set $S$ in the family $\mathcal{F}$. The equation corresponding to $S$ in the formula $\phi$ is $\underset{u\in S}{\Oplus}x_u = 1$. The number of variables in the left-hand side that are set to $1$ by $\sigma_2$ is $|X\cap S|$, which is an odd number. Also, all $|S|$ (again, which is an odd number) variables in the left-hand side are set to $1$ by $\sigma_1$. Therefore, the left-hand side evaluates to $1$ under both $\sigma_1$ and $\sigma_2$. So, both these assignments satisfy the equation $\underset{u\in S}{\Oplus}x_u=1$.

Next, we argue that $\sigma_1$ and $\sigma_2$ satisfy the equations of $\phi$ that were added corresponding to even-sized sets of the family $\mathcal{F}$. Consider any even-sized set $S$ in the family $\mathcal{F}$. The $k+1$ equations corresponding to $S$ in the formula $\phi$ are $y_S\oplus z_S^{1}=0, y_S\oplus z_S^{2}=0, \ldots, y_S\oplus z_S^{k}=0$ and $\underset{u\in S}{\Oplus}x_u\oplus y_S=0$. Consider any of the first $k$ equations, say $y_S\oplus z_S^{i} = 0$, where $1\leq i\leq k$. Both variables on the left-hand side, i.e., $y_S$ and $z_S^{i}$, are assigned the same truth value, i.e., both $0$ by $\sigma_1$ and both $1$ by $\sigma_2$. So, both these assignments satisfy the equation $y_S\oplus z_S^{i}=0$. Next, consider the last equation, i.e., $\underset{u\in S}{\Oplus}x_u \oplus y_S=0$. The number of variables amongst $x_u\bigl\lvert_{u\in S}$ that are set to $1$ by $\sigma_2$ is $|X\cap S|$, which is an odd number. Also, the variable $y_S$ is set to $1$ by $\sigma_2$. Therefore, overall, the number of variables in the left-hand side that are set to $1$ by $\sigma_2$ is even. Also, $\sigma_1$ sets all variables on the left-hand side to $1$ except $y_S$. That is, it sets all the $|S|$ (again, which is an even number) variables $x_u\big\lvert_{u\in S}$ to $1$. Therefore, the left-hand side evaluates to $0$ under both $\sigma_1$ and $\sigma_2$. So, both these assignments satisfy the equation $\underset{u\in S}{\Oplus}x_u\oplus y_S=0$. 

Hence, $(\phi, n-k)$ is a YES instance of \textsc{Exact Differ Affine}-SAT. 

\textbf{Reverse direction}. Suppose that $(\phi, n-k)$ is a YES instance of \textsc{Exact Differ Affine}-SAT. That is, there are satisfying assignments $\sigma_1$ and $\sigma_2$ of $\phi$ that overlap on exactly $k$ variables.  Consider any even-sized set $S$ in the family $\mathcal{F}$. As $\sigma_1$  satisfies the equations $y_S\oplus z_S^{1}=0, y_S\oplus z_S^{2}=0, \ldots, y_S\oplus z_S^{k}=0$,  it must assign the same truth value to all the $k+1$ variables $y_S, z_S^{1}, z_S^{2}, \ldots, z_S^{k}$. Similarly, $\sigma_2$ must assign the same truth value to $y_S, z_S^{1}, z_S^{2}, \ldots, z_S^{k}$. Therefore, either $\sigma_1$ and $\sigma_2$ overlap on all these $k+1$ variables, or they differ on all these $k+1$ variables. So, as there are only $k$ overlaps, $\sigma_1$ and $\sigma_2$ must differ on $y_S, z_S^{1}, z_{S}^2,\ldots, z_{S}^k$. Thus, all the $k$ overlaps occur at $x$ variables. Let $X$ denote the $k$-sized set ${\big\{u\in \mathcal{U}~|~\sigma_1 \mbox{ and } \sigma_2 \mbox{ differ on } x_u\big\}}$. Now, we show that $|X\cap S|$ is odd for all sets $S$ in the family $\mathcal{F}$. 

First, we argue that $X$ has odd-sized intersection with all odd-sized sets of the family $\mathcal{F}$. Consider any odd-sized set $S$ in the family $\mathcal{F}$. The equation corresponding to $S$ in the formula $\phi$ is $\underset{u\in S}{\Oplus}x_u=1$. We split the left-hand side into two parts to express this equation as $\underbracket[0.4pt]{\underset{u\in X\cap S}{\Oplus}x_u}_{\clap{\scriptsize{A}}}\oplus \underbracket[0.4pt]{\underset{u\in S\setminus X}{\Oplus}x_u}_{\clap{\scriptsize{B}}} = 1$. Note that $\sigma_1$ and $\sigma_2$ overlap on all variables in the first part, i.e., $A$. So, $A$ evaluates to the same truth value under both assignments. Thus, as both $\sigma_1$ and $\sigma_2$ satisfy this equation, they must assign the same truth value to the second part, i.e., $B$, as well. Also, $\sigma_1$ and $\sigma_2$ differ on all variables in $B$. So, for its truth value to be same under both assignments, $B$ must have an even number of variables. That is, $|S\setminus X|$ must be even. Now, as $|S|$ is odd and $|S\setminus X|$ is even, we infer that $|X\cap S|=|S|-|S\setminus X|$ is odd. 

Next, we argue that $X$ has odd-sized intersection with all even-sized sets of the family $\mathcal{F}$. Consider any even-sized set $S$ in the family $\mathcal{F}$. Amongst the $k+1$ equations corresponding to $S$ in the formula $\phi$, consider the last equation, i.e., $\underset{u\in S}{\Oplus}x_u \oplus y_S = 0$. We split the left-hand side into two parts to express this equation as  $\underbracket[0.4pt]{\underset{u\in X\cap S}{\Oplus}x_u}_{\clap{\scriptsize{A}}}\oplus \underbracket[0.4pt]{\underset{u\in S\setminus X}{\Oplus}x_u\oplus y_S}_{\clap{\scriptsize{B}}}=0$. Note that $\sigma_1$ and $\sigma_2$ overlap on all variables in the first part, i.e., $A$. So, $A$ evaluates to the same truth value under both assignments. Thus, as both $\sigma_1$ and $\sigma_2$ satisfy this equation, they must assign the same truth value to the second part, i.e., $B$. Also, $\sigma_1$ and $\sigma_2$ differ on all variables in $B$. So, for its truth value to be same under both assignments, $B$ must have an even number of variables. That is, $|S\setminus X|+1$ must be even. Now, as $|S|$ is even and $|S\setminus X|$ is odd, we infer that $|X\cap S| = |S| - |S\setminus X|$ is odd.

Hence, $(\mathcal{U}, \mathcal{F}, k)$ is a YES instance of \textsc{Exact Odd Set}.

This reduction also works with \textsc{Odd Set} as the source problem and \textsc{Max Differ Affine}-SAT as the target problem. So, both \textsc{Exact Differ Affine-SAT} and \textsc{Max Differ Affine-SAT} are \WOH in the parameter $n-d$. This proves \cref{Exact Differ Affine-SAT W[1]-hard in n-d}.
\end{proof}


\section{2-CNF formulas}
\label{sec: 2-CNF formulas}
In this section, we explore the classical and parameterized complexity of \textsc{Max Differ 2-SAT} and \textsc{Exact Differ 2-SAT}. We first show that these problems are polynomial time solvable on $(2,2)$-CNF formulas by constructing a graph corresponding to the instance and observing some structural properties of that graph. Then we show that both of these problems are \WOH with respect to the parameter $d$. We begin by proving the following theorem. 

\begin{theorem} \label{Max Differ 2-SAT poly-time (2 2)-CNF formulas}
\textsc{Max Differ 2-SAT} is polynomial-time solvable on $(2,2)$-CNF formulas.
\end{theorem}
\begin{proof}
Consider an instance $(\phi,d)$ of \textsc{Max Differ 2-SAT}, where $\phi$ is a $(2,2)$-CNF formula on $n$ variables. We construct a graph $G$ as follows: For every variable $x$ of $\phi$, introduce vertices $x[0]$ and $x[1]$. We add an edge corresponding to each clause of $\phi$ as follows: For every clause of the form $x\vee y$, add the edge $\{x[1],y[1]\}$. For every clause of the form $\neg x\vee \neg y$, add the edge $\{x[0],y[0]\}$. For every clause of the form $x\vee \neg y$, add the edge $\{x[1],y[0]\}$. We refer to such edges as \emph{clause-edges}. Also, for every variable $x$ of $\phi$, add the edge $\{x[0], x[1]\}$. We refer to such edges as \emph{matching-edges}. See \cref{fig:graph for Max Differ 2-SAT} for an example. 

The structure of $G$ can be understood as follows: In the absence of matching-edges, the graph is a disjoint union of some paths and cycles. Some of these paths/cycles get glued together upon the addition of matching-edges. Observe that a matching-edge can either join a terminal vertex of one path with a terminal vertex of another path, or join any vertex of a path/cycle to an isolated vertex. These are the only two possibilities because the endpoints of any matching-edge together can have only at most two neighbors (apart from each other). Thus, once the matching-edges are added, observe that each component of $G$ becomes a path or a cycle, possibly with pendant matching-edges attached to some of its vertices. We call them \emph{path-like} and \emph{cycle-like} components respectively. We further classify a cycle-like component based on the parity of the length of its corresponding cycle: \emph{odd-cycle-like} and \emph{even-cycle-like}. For each of these three types of components, we analyze the maximum number of its variables on which any two assignments can differ such that both of them satisfy all clauses corresponding to its clause-edges. Consider any component, say $C$, of $G$.

\begin{figure}
\begin{center}
\includegraphics[scale=0.7]{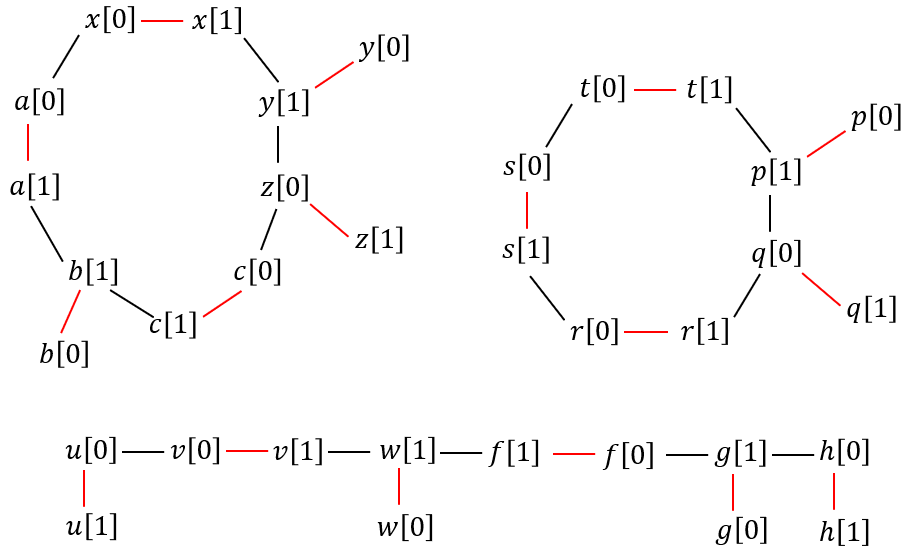}
\end{center}
\caption{This figure shows the graph $G$ constructed in the proof of \cref{Max Differ 2-SAT poly-time (2 2)-CNF formulas} \color{black} when the input $(2,2)$-CNF formula $\phi$ is $(x\vee y)\wedge (y\vee \neg z)\wedge (\neg z\vee \neg c)\wedge (c\vee b) \wedge (b\vee a)\wedge (\neg a\vee \neg x) \wedge (t \vee p) \wedge$ $(p \vee \neg q) \wedge (\neg q \vee r) \wedge (\neg r\vee s) \wedge (\neg s\vee \neg t) \wedge (\neg u \vee \neg v) \wedge (v \vee w) \wedge (w \vee f) \wedge (\neg f \vee g) \wedge (g \vee \neg h)$. The red line segments denote the matching-edges, and the black line segments denote the clause-edges. Note that $G$ has one odd-cycle-like component, one even-cycle-like component and one path-like component.}
\label{fig:graph for Max Differ 2-SAT}
\end{figure}
\textbf{Case 1}. $C$ is a path-like component.

Starting from one of the two terminal vertices of the corresponding path, i) let $x_1[i_1]$, $x_2[i_2]$, $x_3[i_3]$ $\ldots$ denote the odd-positioned vertices, and ii) let $y_1[j_1]$, $y_2[j_2]$, $y_3[j_3]$, $\ldots$ denote the even-positioned vertices. That is, the path has these vertices in the following order: $x_1[i_1]$, $y_1[j_1]$, $x_2[i_2]$, $y_2[j_2]$, $x_3[i_3]$, $y_3[j_3]$, $\ldots$ and so on. 

Let $\sigma_1$ denote the following assignment: i) $\sigma_1$ sets $x_1$ to $i_1$, $x_2$ to $i_2$, $x_3$ to $i_3$, $\ldots$ and so on. ii) For those $y_{k}[j_k]$'s $\big($amongst $y_1[j_1]$, $y_2[j_2]$, $y_3[j_3]$, $\ldots\big)$ whose neighbor along the matching-edge $\big($i.e., $y_k[1-j_{k}]\big)$ is not one of its two neighbors along the path $\big($i.e., $x_{k}[i_k]$ and $x_{k+1}[i_{k+1}]\big)$, and rather hangs as a pendant vertex attached to $y_k[j_k]$: $\sigma_1$ sets $y_k$ to $1-j_k$. Similarly, let $\sigma_2$ denote the following assignment: i) $\sigma_2$ sets $y_1$ to $j_1$, $y_2$ to $j_2$, $y_3$ to $j_3$, $\ldots$ and so on. ii) For those $x_k[i_k]$'s $\big($amongst $x_1[i_1]$, $x_2[i_2]$, $x_3[i_3]$, $\ldots\big)$ whose neighbor along the matching-edge $\big($i.e., $x_k[1-i_k]\big)$ is not one of its two neighbors along the path $\big($i.e., $y_{k-1}[j_{k-1}]$ and $y_{k}[j_k]\big)$, and rather hangs as a pendant vertex attached to $x_k[i_k]$: $\sigma_2$ sets $x_k$ to $1-i_k$. 

Observe that any clause-edge of $C$ is of the form $\big\{x_{k}[i_{k}], y_{k}[j_{k}]\big\}$ or $\big\{y_{k}[j_{k}], x_{k+1}[i_{k+1}]\big\}$. In either case, both $\sigma_1$ and $\sigma_2$ satisfy the corresponding clause. This is because i) in the first case, $\sigma_1$ sets $x_{k}$ to $i_{k}$ and $\sigma_2$ sets $y_{k}$ to $j_{k}$, and ii) in the second case, $\sigma_1$ sets $x_{k+1}$ to $i_{k+1}$ and $\sigma_2$ sets $y_{k}$ to $j_{k}$. Thus, overall, both $\sigma_1$ and $\sigma_2$ satisfy all clauses corresponding to the clause-edges of $C$. Also, $\sigma_1$ and $\sigma_2$ differ on all variables that appear in $C$ $\big($i.e., $x_1, x_2, x_3, \ldots$ and $y_1, y_2, y_3, \ldots\big)$. 

As an example, for the path-like component shown in \cref{fig:graph for Max Differ 2-SAT}, starting from the terminal vertex $u[0]$, i) $\sigma_1$ sets $u$ to $0$, $v$ to $1$, $w$ to $0$, $f$ to $1$, $g$ to $1$, $h$ to $1$, and ii) $\sigma_2$ sets $u$ to $1$, $v$ to $0$, $w$ to $1$, $f$ to $0$, $g$ to $0$, $h$ to $0$. 

\textbf{Case 2}. $C$ is an even-cycle-like component.

Starting from an arbitrary vertex of the corresponding cycle, and moving along the cycle in one direction (say, clockwise): i) let $x_1[i_1]$, $x_2[i_2]$, $\ldots$, $x_{\ell}[i_{\ell}]$ denote the odd-positioned vertices, and ii) let $y_1[j_1]$, $y_2[j_2]$, $\ldots$, $y_{\ell}[j_{\ell}]$ denote the even-positioned vertices. That is, the cycle has these vertices in the following order: $x_1[i_1]$, $y_1[j_1]$, $x_2[i_2]$, $y_2[j_2]$, $\ldots$, $x_{\ell}[i_{\ell}]$, $y_{\ell}[j_{\ell}]$. 

Let $\sigma_1$ denote the following assignment: i) $\sigma_1$ sets $x_1$ to $i_1$, $x_2$ to $i_2$, $\ldots$, $x_{\ell}$ to $i_{\ell}$. ii) For those $y_k[j_k]$'s $\big($amongst $y_1[j_1]$, $y_2[j_2]$, $\ldots$, $y_{\ell}[j_{\ell}]\big)$ whose neighbor along the matching-edge $\big($i.e., $y_k[1-j_k]\big)$ is not one of its two neighbors along the cycle $\big($i.e.,$x_{k}\big[i_{k}\big]$ and $x_{k+1}[i_{k+1}]$\footnote{When $k=\ell$, take $k+1$ as $1$.}$\big)$, and rather hangs as a pendant vertex attached to $y_k[j_k]$: $\sigma_1$ sets $y_k$ to $1-j_k$. Similarly, let $\sigma_2$ denote the following assignment: i) $\sigma_2$ sets $y_1$ to $j_1$, $y_2$ to $j_2$, $\ldots$, $y_{\ell}$ to $j_{\ell}$. ii) For those $x_k[i_k]$'s $\big($amongst $x_1[i_1]$, $x_2[i_2]$, $\ldots$, $x_{\ell}[i_{\ell}]\big)$ whose neighbor along the matching-edge $\big($i.e., $x_k[1-i_k]\big)$ is not one of its two neighbors along the cycle $\big($i.e., $y_{k-1}[j_{k-1}]$\footnote{When $k=1$, take $k-1$ as $\ell$.} and $y_{k}[j_k]\big)$, and rather hangs as a pendant vertex attached to $x_k[i_k]$: $\sigma_2$ sets $x_k$ to $1-i_k$. Using the same argument as in Case 1, both $\sigma_1$ and $\sigma_2$ satisfy all clauses corresponding to the clause-edges of $C$, and they differ on all variables that appear in $C$ $\big($i.e., $x_1$, $x_2$, $\ldots$, $x_{\ell}$ and $y_1$, $y_2$, $\ldots$, $y_{\ell}\big)$. 

As an example, for the even-cycle-like component shown in \cref{fig:graph for Max Differ 2-SAT}, starting from the vertex $q[0]$, i) $\sigma_1$ sets $q$ to $0$, $r$ to $0$, $s$ to $0$, $t$ to $1$, $p$ to $0$, and ii) $\sigma_2$ sets $r$ to $1$, $s$ to $1$, $t$ to $0$, $p$ to $1$, $q$ to $1$.

\textbf{Case 3}. $C$ is an odd-cycle-like component.

If every vertex $v[i]$ of the corresponding cycle is such that its neighbor along the matching-edge $\big($i.e., $v[1-i]\big)$ also belongs to the cycle, then the cycle would have even-length. So, as the cycle has odd-length, it must have at least one vertex $v[i]$ such that $v[1-i]$ is not one of its two neighbors along the cycle, and rather hangs as a pendant vertex attached to $v[i]$. Then, we construct assignments $\sigma_1$ and $\sigma_2$ as follows: Both $\sigma_1$ and $\sigma_2$ set $v$ to $i$, thereby satisfying the two clauses that correspond to the two clause-edges that are incident to $v[i]$ along the cycle. Then, since $C\setminus \big\{v[i], v[1-i]\big\}$ is a path-like component, we extend the assignments $\sigma_1$ and $\sigma_2$ to the remaining variables of $C$ $\big($i.e., other than $v\big)$ in the manner described in Case 1. This ensures that both $\sigma_1$ and $\sigma_2$ satisfy all clauses corresponding to the clause-edges of $C$. Also, they differ on all but one variable $\big($i.e., all except $v\big)$ of $C$. 

As an example, for the odd-cycle-like component shown in \cref{fig:graph for Max Differ 2-SAT}, after setting $y$ to $1$ in both $\sigma_1$ and $\sigma_2 $, we remove $y[0]$ and $y[1]$ from the component (see \cref{fig: odd-cycle-like comp to path-like comp}). Then, starting from the terminal vertex $z[0]$ of the path corresponding to the path-like component so obtained, i) $\sigma_1$ sets $z$ to $0$, $c$ to $1$, $b$ to $0$, $a$ to $1$, $x$ to $0$, and ii) $\sigma_2$ sets $z$ to $1$, $c$ to $0$, $b$ to $1$, $a$ to $0$, $x$ to $1$.

\begin{figure}
\includegraphics[scale=0.7]{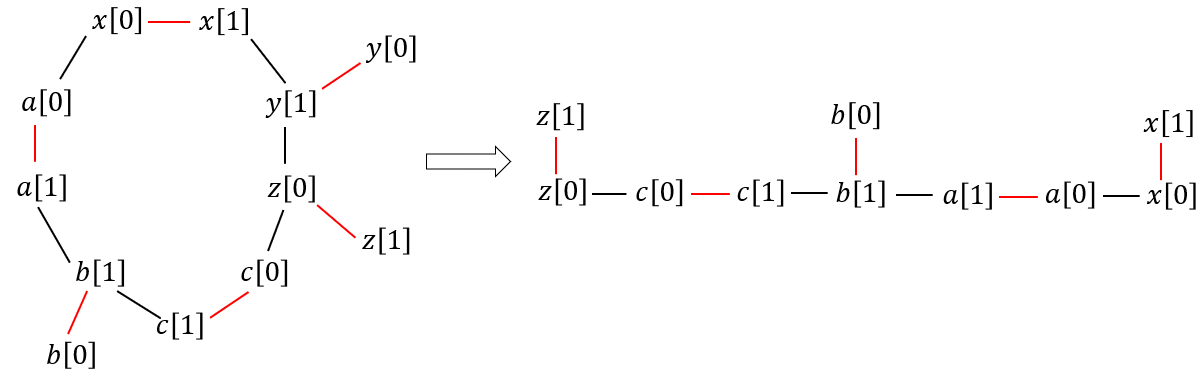}
\caption{This figure is with reference to Case 3 in the proof of \cref{Max Differ 2-SAT poly-time (2 2)-CNF formulas}. Note that the removal of  $y[0]$ and $y[1]$ from the odd-cycle-like component shown in \cref{fig:graph for Max Differ 2-SAT} results in a path-like component.}
\label{fig: odd-cycle-like comp to path-like comp}
\end{figure}
Next, we show that at least one overlap is unavoidable no matter how the variables of $C$ are assigned truth values by a pair of satisfying assignments. For the sake of contradiction, assume that there are assignments $\sigma_1$ and $\sigma_2$ that satisfy all clauses corresponding to the clause-edges of $C$, and also differ on all variables of $C$. Then, construct vertex-subsets $S_1$ and $S_2$ of $C$ as follows: For every variable $x$ that appears in $C$, i) if $\sigma_1$ sets $x$ to $0$ $\big($and so, $\sigma_2$ sets $x$ to $1\big)$, then add $x[0]$ to $S_1$ and $x[1]$ to $S_2$, and ii) if $\sigma_1$ sets $x$ to $1$ $\big($and so, $\sigma_2$ sets $x$ to $0\big)$, then add $x[1]$ to $S_1$ and $x[0]$ to $S_2$. Observe that $S_1$ and $S_2$ are disjoint vertex covers of $C$. So, $C$ must be a bipartite graph. However, this is not possible as $C$ contains an odd-length cycle. 

Thus, overall, based on Cases 1, 2, 3, we express the maximum number of variables on which any two satisfying assignments of $\phi$ can differ as follows:
\begin{equation*}
\underset{\substack{C:~C \text{ \footnotesize{is a path-like}}\\ \text{~~~~~~~~\footnotesize{component of }} G}}{\sum}|vars(C)| ~~~+ \underset{\substack{C:~C \text{ \footnotesize{is an even-cycle-like}}\\ \text{\footnotesize{component of }} G}}{\sum}|vars(C)| ~~~+ \underset{\substack{C:~C \text{ \footnotesize{is an odd-cycle-like}}\\ \text{\footnotesize{component of }} G}}{\sum}\big(|vars(C)|-1\big)
\end{equation*}
\begin{equation*}
= n - \big|\big\{C~|~C \text{ is an odd-cycle-like component of } G\big\}\big|,
\end{equation*}
where $vars(C)$ denotes the set of variables that appear in the component $C$. So, if the number of odd-cycle-like components in $G$ is at most $n-d$, we return YES; otherwise, we return NO. This proves \cref{Max Differ 2-SAT poly-time (2 2)-CNF formulas}.
\end{proof}

We use similar ideas in the proof of \cref{Max Differ 2-SAT poly-time (2 2)-CNF formulas} to show that \textsc{Exact Differ 2-SAT} can also be solved in polynomial time on $(2,2)$-CNF formulas. This requires more careful analysis of the graph constructed and a reduction to \textsc{Subset Sum problem}, as we want the individual contributions, in terms of number of variables where the assignments differ, to sum up to an exact value. We show the result in the following theorem. 

\begin{theorem}
\label{Exact Differ 2-SAT poly-time on (2 2)-CNF formulas} \textsc{Exact Differ 2-SAT} is polynomial-time solvable on $(2,2)$-CNF formulas.
\end{theorem}
\begin{proof}
Consider an instance $(\phi,d)$ of \textsc{Exact Differ 2-SAT}, where $\phi$ is a $(2,2)$-CNF formula. We re-consider the graph $G$ constructed in the proof of \cref{Max Differ 2-SAT poly-time (2 2)-CNF formulas}. Recall that $G$ has three types of components, namely path-like components, even-cycle-like components and odd-cycle-like components. We further classify the even-cycle-like components into two categories: ones with pendants, and ones without pendants. That is, the former category consists of those even-cycle-like components whose corresponding cycles contain at least one vertex that has a pendant edge attached to it, while the latter category consists of all even-cycle components. For each of the four types of components, we analyze the possible values for the number of its variables on which any two assignments can differ such that both of them satisfy all clauses corresponding to its clause-edges. Consider any component, say $C$, of $G$. We use $vars(C)$ to denote the set of variables that appear in $C$.

\textbf{Case 1}. $C$ is a path-like component.

Consider any $0\leq \mu\leq |vars(C)|$. We move along the corresponding path, starting from one of its two terminal vertices. At any point, we say that a variable $v$ has been \emph{seen} if either i) both $v[0]$ and $v[1]$ are vertices of the sub-path traversed so far, or ii) one of $v[0]$ and $v[1]$ is a vertex of the sub-path traversed so far, and the other one hangs as a pendant vertex attached to it. We stop moving along the path once the first $\mu$ variables have been seen. Let $x[i]$ denote the path's vertex at which we stop, and let $y[j]$ denote its immediate next neighbor along the path. Observe that the edge joining $x[i]$ and $y[j]$ along the path is a clause-edge, whose removal breaks $C$ into two smaller path-like components, say $C_1$ $\big($containing $x\big)$ and $C_2$ $\big($containing $y\big)$. 

We construct assignments $\sigma_1$ and $\sigma_2$ as follows: For the path-like component $C_1$, we re-apply the analysis of Case 1 in the proof of \cref{Max Differ 2-SAT poly-time (2 2)-CNF formulas}. That is, as described therein, $\sigma_1$ and $\sigma_2$ are made to set the variables of $C_1$ such that both of them satisfy all clauses corresponding to the clause-edges of $C_1$, and they differ on all its $|vars(C_1)|$ many variables. Then, for the path-like component $C_2$, both $\sigma_1$ and $\sigma_2$ are made to mimic one of the two assignments $\big($namely, the one that sets $y$ to $j\big)$ constructed in Case 1 in the proof of \cref{Max Differ 2-SAT poly-time (2 2)-CNF formulas}. This ensures that both $\sigma_1$ and $\sigma_2$ satisfy all clauses corresponding to the clause-edges of $C_2$, and they overlap on all its $|vars(C_2)|$ many variables. Further, since both $\sigma_1$ and $\sigma_2$ set $y$ to $j$, they also satisfy the clause corresponding to the clause-edge $\{x[i], y[j]\}$. Thus, overall, both $\sigma_1$ and $\sigma_2$ satisfy all clauses corresponding to the clause-edges of $C$, and they differ on exactly $|vars(C_1)|=\mu$ many of its variables. 

For example, consider the path-like component shown in \cref{fig:graph for Max Differ 2-SAT}, and let $\mu=3$. We traverse the corresponding path starting from the terminal vertex $u[0]$, and we stop at $w[1]$, i.e., once three variables have been seen. Now, removing the clause-edge $\{w[1], f[1]\}$ results in two path-like components: one on variables $u$, $v$, $w$, and the other on variables $f$, $g$, $h$ $\big($see \cref{fig: path-like comp to two small path-like comps}$\big)$. Then, i) $\sigma_1$ sets $u$ to $0$, $v$ to $1$, $w$ to $0$, ii) $\sigma_2$ sets $u$ to $1$, $v$ to $0$, $w$ to $1$, and iii) both $\sigma_1$ and $\sigma_2$ set $f$ to $1$, $g$ to $1$, $h$ to $1$. So, $\sigma_1$ and $\sigma_2$ differ on $u$, $v$, $w$ and overlap on $f$, $g$, $h$.

\begin{figure}
\begin{center}
\includegraphics[scale=0.7]{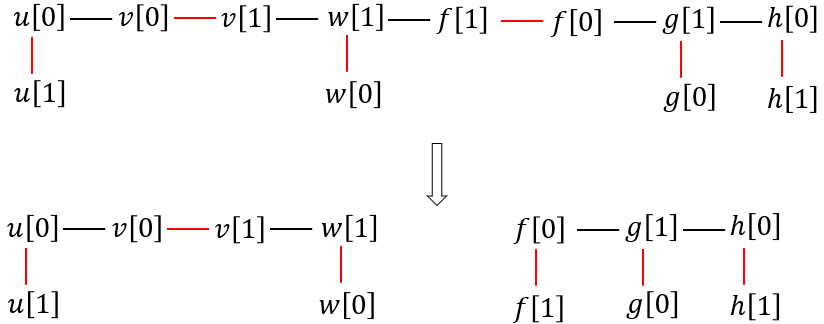}
\end{center}
\caption{This figure is with reference to Case 1 in the proof of \cref{Exact Differ 2-SAT poly-time on (2 2)-CNF formulas}.
Note that the removal of the clause-edge $\{w[1], f[1]\}$ from the path-like component shown in \cref{fig:graph for Max Differ 2-SAT} breaks it into two smaller path-like components.}
\label{fig: path-like comp to two small path-like comps}
\end{figure}

\textbf{Case 2}. $C$ is an even-cycle-like component with pendants.

As argued in Case 2 in the proof of \cref{Max Differ 2-SAT poly-time (2 2)-CNF formulas}, we already know that there is a pair of assignments such that both of them satisfy all clauses corresponding to the clause-edges of $C$, and they differ on all its $|vars(C)|$ many variables. Next, consider any $0\leq \mu\leq |vars(C)|-1$. Arbitrarily fix a vertex $v[i]$ on the cycle corresponding to $C$ such that its neighbor along the matching-edge $\big($i.e., $v[1-i]\big)$ does not belong to the cycle, and rather hangs as a pendant vertex attached to $v[i]$. Then, we construct assignments $\sigma_1$ and $\sigma_2$ as follows: Both $\sigma_1$ and $\sigma_2$ set $v$ to $i$, thereby satisfying the two clauses that correspond to the two clause-edges incident to $v[i]$ along the cycle. Then, since $C\setminus \big\{v[i], v[1-i]\big\}$ is a path-like component, we apply the analysis of Case 1 on it $\big($with the same $\mu\big)$ to assign truth-values to the variables in $vars(C)\setminus \{v\}$ under $\sigma_1$ and $\sigma_2$. Note that both $\sigma_1$ and $\sigma_2$ satisfy all clauses corresponding to the clause-edges of $C$. Also, they overlap on $v$, and they differ on exactly $\mu$ many variables in $vars(C)\setminus \{v\}$.

\textbf{Case 3}. $C$ is an odd-cycle-like component.

As argued in Case 3 in the proof of \cref{Max Differ 2-SAT poly-time (2 2)-CNF formulas}, we already know that at least one overlap is unavoidable amongst the variables in $vars(C)$ for any pair of assignments that satisfy the clauses corresponding to the clause-edges of $C$. Next, consider any $0\leq \mu\leq |vars(C)|-1$. Since the cycle corresponding to $C$ has odd-length, it must contain a vertex $v[i]$ whose neighbor $v[1-i]$ hangs as a pendant vertex attached to it. Then, using the same argument as that in Case 2, we obtain a pair of assignments such that both of them satisfy all clauses corresponding to the clause-edges of $C$, and they differ on exactly $\mu$ many variables in $vars(C)$. 

\textbf{Case 4}. $C$ is an even-length cycle.

Observe that the vertices in $C$, in order of their appearance along the cycle, must be of the following form: $x_1[i_1]$, $x_1[1-i_1]$, $x_2[i_2]$, $x_2[1-i_2]$, $\ldots$, $x_{\ell}[i_{\ell}]$, $x_{\ell}[1-i_{\ell}]$ $\big($see \cref{fig: even-length cycle}$\big)$. Here, $x_1$, $x_2$, $\ldots$, $x_{\ell}$ are variables, and $i_1, \ldots, i_{\ell}\in\{0,1\}$. Note that setting $x_1$ to $i_1$ forces $x_2$ to be set to $i_2$ $\big($in order to satisfy the clause corresponding to the clause-edge $\{x_1[1-i_1], x_2[i_2]\}\big)$, which in turn forces $x_3$ to be set to $i_3$ $\big($in order to satisfy the clause corresponding to the clause-edge $\{x_2[1-i_2], x_3[i_3]\}\big)$, $\ldots$ and so on. Similarly, setting $x_1$ to $1-i_1$ forces $x_{\ell}$ to be set to $1-i_{\ell}$ $\big($in order to satisfy the clause corresponding to the clause-edge $\{x_{\ell}[1-i_{\ell}], x_1[i_1]\}\big)$, which in turn forces $x_{\ell-1}$ to be set to $1-i_{\ell-1}$ $\big($in order to satisfy the clause corresponding to the clause-edge $\{x_{\ell-1}[1-i_{\ell-1}], x_{\ell}[i_{\ell}]\}\big)$, $\ldots$ and so on. Therefore, it is clear that there are only two different assignments that satisfy the clauses corresponding to the $\ell$ clause-edges of $C$: i) one that sets $x_1$ to $i_1$, $x_2$ to $i_2$, $\ldots$, $x_{\ell}$ to $i_{\ell}$, and ii) the other that sets $x_1$ to $1-i_1$, $x_2$ to $1-i_2$, $\ldots$, $x_{\ell}$ to $1-i_{\ell}$. So, any pair of such assignments either i) differ on all $|vars(C)|$ many variables of $C$, or ii) overlap on all $|vars(C)|$ many variables of $C$. 

\begin{figure}
\begin{center}
\includegraphics[scale=0.7]{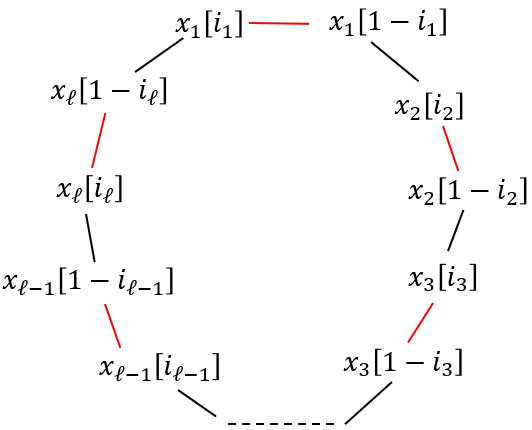}
\end{center}
\caption{This figure shows the even-length cycle $C$ in Case 4 in the proof of \cref{Exact Differ 2-SAT poly-time on (2 2)-CNF formulas}.}
\label{fig: even-length cycle}
\end{figure}
Thus, overall, based on Cases 1, 2, 3, 4, we know the following: i) Amongst the variables of all path-like components, even-cycle-like components with pendants and odd-cycle-like components of $G$, the number $d'$ of variables on which any two satisfying assignments can differ is allowed to take any value ranging from $0$ to $D$ (both inclusive), where
\begin{equation*}
D:= \underset{\substack{C:~C \text{ \footnotesize{is a path-like component, or an}}\\ \text{\footnotesize{even-cycle-like component with pendants}}}}{\sum}|vars(C)| ~~~ + \underset{\substack{C:~C \text{ \footnotesize{is an odd-cycle-}}\\ \text{\footnotesize{~~~~~~like component~~~~~}}}}{\sum}\big(|vars(C)|-1\big),
\end{equation*}
and ii) amongst the variables of all even-cycle components of $G$ $\big($say $C_1,C_2, \ldots, C_t\big)$, the number of variables on which any two satisfying assignments can differ is allowed to take any value that can be obtained as a sum of some elements of $\big\{|vars(C_1)|$, $|vars(C_2)|$, $\ldots$, $|vars(C_t)|\}$. So, for each $0\leq d'\leq D$, we run the algorithm for \textsc{Subset Sum problem} on the multi-set $\big\{|vars(C_1)|$, $|vars(C_2)|$, $\ldots$, $|vars(C_t)|\big\}$ with $d-d'$ as the target sum. We return YES if at least one of these runs returns YES; otherwise, we return NO. This proves \cref{Exact Differ 2-SAT poly-time on (2 2)-CNF formulas}.
\end{proof}

Looking at the parameterized complexity of \textsc{Exact Differ 2-SAT} and \textsc{Max Differ 2}-SAT with respect to the parameter $d$, we next prove the following hardness result. 
\begin{theorem}
\label{Exact Differ 2-SAT W[1]-hard in d} \textsc{Exact/Max Differ 2-SAT} is \WOH in the parameter $d$. 
\end{theorem}
\begin{proof}
We describe a reduction from \textsc{Independent Set}. Consider an instance $(G,k)$ of \textsc{Independent Set}. We construct a 2-CNF formula $\phi$ as follows: For every vertex $v\in V(G)$, introduce two variables $x_v$ and $y_v$; we refer to them as $x$-variable and $y$-variable respectively. For every edge $uv\in E(G)$, i) we add a clause that consists of the $x$-variables corresponding to the vertices $u$ and $v$, i.e., $x_u \vee x_v$, and ii) we add a clause that consists of the $y$-variables corresponding to the vertices $u$ and $v$, i.e., $y_u\vee y_v$. For every pair of vertices $u,v \in V(G)$, we add a clause that consists of the $x$-variable corresponding to $u$ and the $y$-variable corresponding to $v$, i.e., $x_u \vee y_v$. We set $d=2k$. 

We prove that $(G,k)$ is a YES instance of \textsc{Independent Set} if and only if $(\phi, d)$ is a YES instance of \textsc{Exact Differ 2-SAT}. At a high level, we argue this equivalence as follows: In the forward direction, we show that the two desired satisfying assignments are i) the assignment that assigns $0$ to all $x$-variables corresponding to the vertices of the given independent set, and $1$ to the remaining variables, and ii) the assignment that assigns $0$ to all $y$-variables corresponding to the vertices of the given independent set, and $1$ to the remaining variables. In the reverse direction, we partition the set of variables on which the two given assignments differ into two parts: i) one part consists of those variables that are set to $1$ by the first assignment, and $0$ by the second assignment, and ii) the other part consists of those variables that are set to $0$ by the first assignment, and $1$ by the second assignment. Then, we show that at least one of these two parts has the desired size, and it is not a mix of $x$-variables and $y$-variables. That is, either it has only $x$-variables, or it has only $y$-variables. Finally, we show that the  vertices that correspond to the variables in this part form the desired independent set. We make this argument precise below.

\textbf{Forward direction}. Suppose that $(G,k)$ is a YES instance of \textsc{Independent Set}. That is, $G$ has a $k$-sized independent set $X$. Let $\sigma_1$ and $\sigma_2$ be assignments of $\phi$ defined as follows: For every vertex $v\in X$, i) $\sigma_1$ sets $x_v$ to $0$, $\sigma_2$ sets $x_v$ to $1$, and ii) $\sigma_1$ sets $y_v$ to $1$, $\sigma_2$ sets $y_v$ to $0$. For every vertex $v\in V(G)\setminus X$, both $\sigma_1$ and $\sigma_2$ set $x_v$ and $y_v$ to $1$. Note that $\sigma_1$ and $\sigma_2$ differ on $2|X|=2k$ variables, namely $x_v\big\lvert_{v\in X}$ and $y_v\big\lvert_{v\in X}$. Now, we show that $\sigma_1$ and $\sigma_2$ are satisfying assignments of $\phi$.

First, we argue that $\sigma_1$ and $\sigma_2$ satisfy the clauses that were added corresponding to the edges of $G$. Consider an edge $uv\in E(G)$. The clauses added in $\phi$ corresponding to this edge are $x_u\vee x_v$ and $y_u\vee y_v$. Since $X$ is an independent set, no edge of $G$ has both its endpoints inside $X$. So, at least one endpoint of the edge $uv$ must be outside $X$. Without loss of generality, assume that $u\in V(G)\setminus X$. Then, both $\sigma_1$ and $\sigma_2$ set $x_u$ and $y_u$ to $1$, thereby satisfying the clauses $x_u\vee x_v$ and $y_u\vee y_v$. 

Next, we argue that $\sigma_1$ and $\sigma_2$ satisfy the clauses that were added corresponding to the vertex-pairs of $G$. Consider a pair of vertices $u,v\in V(G)$. The clause added in $\phi$ corresponding to this vertex-pair is $x_u\vee y_v$. Suppose that at least one of $u$ and $v$ lies outside $X$. If $u\in V(G)\setminus X$, then both $\sigma_1$ and $\sigma_2$ set $x_u$ to $1$. Similarly, if $v\in V(G)\setminus X$, then both $\sigma_1$ and $\sigma_2$ set $y_v$ to $1$. Therefore, in either case, $\sigma_1$ and $\sigma_2$ satisfy the clause $x_u\vee y_v$. Now, assume that both $u$ and $v$ lie inside $X$. Then, since $\sigma_1$ sets $y_v$ to $1$ and $\sigma_2$ sets $x_u$ to $1$, both these assignments satisfy the clause $x_u\vee y_v$. 

Hence, $(\phi, 2k)$ is a YES instance of \textsc{Exact Differ 2-SAT}.

\textbf{Reverse direction}. Suppose that $(\phi, 2k)$ is a YES instance of \textsc{Exact Differ 2-SAT}. That is, there are satisfying assignments $\sigma_1$ and $\sigma_2$ of $\phi$ that differ on $2k$ variables. We partition the set of variables on which $\sigma_1$ and $\sigma_2$ differ as $D_1\uplus D_2$, where i) $D_1$ denotes the set of those variables that are set to $1$ by $\sigma_1$ and $0$ by $\sigma_2$, and ii) $D_2$ denotes the set of those variables that are set to $0$ by $\sigma_1$ and $1$ by $\sigma_2$. Since $|D_1|+|D_2|= 2k$, at least one of $D_1$ and $D_2$ has size at least $k$. Without loss of generality, assume that $|D_1|\geq k$. 

Now, we argue that either all variables in $D_1$ are $x$-variables, or all variables in $D_1$ are $y$-variables. For the sake of contradiction, assume that $x_u\in D_1$ and $y_v\in D_1$ for some vertices $u,v \in V(G)$. That is, i) $\sigma_1$ sets both $x_u$ and $y_v$ to $1$, and ii) $\sigma_2$ sets both $x_u$ and $y_v$ to $0$. Then, $\sigma_2$ fails to satisfy the clause $x_u\vee y_v$ added in $\phi$ corresponding to the vertex-pair $u, v$. However, this is not possible because $\sigma_2$ is given to be a satisfying assignment of $\phi$. 

First, suppose that all variables in $D_1$ are $x$-variables. Let $X:=\big\{u\in V(G)~|~x_u\in D_1\big\}$. Note that $|X|=|D_1|\geq k$. We show that $X$ is an independent set of $G$. For the sake of contradiction, assume that there is an edge $uv\in E(G)$ with both endpoints in $X$. That is, i) $\sigma_1$ sets both $x_u$ and $x_v$ to $1$, and ii) $\sigma_2$ sets both $x_u$ and $x_v$ to $0$. Then, $\sigma_2$ fails to satisfy the clause $x_u\vee x_v$ added in $\phi$ corresponding to the edge $uv$. However, this is not possible because $\sigma_2$ is given to be a satisfying assignment of $\phi$.   

Next, suppose that all variables in $D_1$ are $y$-variables. Let $Y:=\big\{u\in V(G)~|~y_u\in D_1\big\}$. Note that $|Y|=|D_1|\geq k$. We show that $Y$ is an independent set of $G$. For the sake of contradiction, assume that there is an edge $uv\in E(G)$ with both endpoints in $Y$. That is, i) $\sigma_1$ sets both $y_u$ and $y_v$ to $1$, and ii) $\sigma_2$ sets both $y_u$ and $y_v$ to $0$. Then, $\sigma_2$ fails to satisfy the clause $y_u\vee y_v$ added in $\phi$ corresponding to the edge $uv$. However, this is not possible because $\sigma_2$ is given to be a satisfying assignment of $\phi$. 

Hence, $(G,k)$ is a YES instance of \textsc{Independent Set}.

This reduction also works with \textsc{Max Differ 2}-SAT as the target problem. So, both \textsc{Exact Differ 2-SAT} and \textsc{Max Differ 2-SAT} are \WOH in the parameter $d$. This proves \cref{Exact Differ 2-SAT W[1]-hard in d}.
\end{proof}

\section{Hitting formulas}
\label{sec: hitting formulas}
In this section, we consider hitting formulas, and we show that both its diverse variants, i.e., \textsc{Exact Differ Hitting-SAT} and \textsc{Max Differ Hitting-SAT}, are polynomial-time solvable. 
\begin{theorem}
\label{Exact Differ Hitting-SAT poly-time}
\textsc{Exact/Max Differ Hitting-SAT} admits a polynomial-time algorithm.
\end{theorem}
\begin{proof}
Consider an instance $(\phi,d)$ of \textsc{Exact Differ Hitting-SAT}, where $\phi$ is a hitting formula with $m$ clauses (say $C_1,\ldots, C_m$) on $n$ variables. Let $vars(\phi)$ denote the set of all $n$ variables of $\phi$. For every $1\leq i\leq m$, let $vars(C_i)$ denote the set of all variables that appear in the clause $C_i$. For every $1\leq i,j\leq m$, let $\lambda(i,j)$ denote the number of variables $x\in vars(C_i)\cap vars(C_j)$ such that $x$ appears as a positive literal in one clause, and as a negative literal in the other clause. 

Note that 
\begin{align*}
& \big|\big\{(\sigma_1,\sigma_2)~|~\sigma_1 \mbox{ and } \sigma_2 \mbox{ differ on } d \mbox{ variables, and both } \sigma_1 \mbox{ and } \sigma_2 \mbox{ satisfy } \phi\big\}\big|\\
& = \big|\big\{(\sigma_1,\sigma_2)~|~\sigma_1 \mbox{ and } \sigma_2 \mbox{ differ on } d \mbox{ variables}\big\}\big|\\
& \hspace{0.3 cm} - \big|\big\{(\sigma_1,\sigma_2)~|~\sigma_1 \mbox{ and } \sigma_2 \mbox{ differ on } d \mbox{ variables, and } \sigma_1 \mbox{ falsifies } \phi\big\}\big| \\
& \hspace{0.3 cm} - \big|\big\{(\sigma_1,\sigma_2)~|~\sigma_1 \mbox{ and } \sigma_2 \mbox{ differ on } d \mbox{ variables, and } \sigma_2 \mbox{ falsifies } \phi\big\}\big| \\
& \hspace{0.3 cm} + \big|\big\{(\sigma_1,\sigma_2)~|~\sigma_1 \mbox{ and } \sigma_2 \mbox{ differ on } d \mbox{ variables, and both } \sigma_1 \mbox{ and } \sigma_2 \mbox{ falsify } \phi\big\}\big| \\
& = 2^n\cdot \binom{n}{d} 
- \big|\big\{\sigma_1~|~\sigma_1 \mbox{ falsifies } \phi\big\}\big|\cdot \binom{n}{d} - \big|\big\{\sigma_2~|~\sigma_2 \mbox{ falsifies } \phi\big\}\big|\cdot \binom{n}{d}\\
& \hspace{0.3 cm}
+ \sum_{i=1}^m\sum_{j=1}^m \underbrace{\big|\big\{(\sigma_1,\sigma_2)~|~\sigma_1 \mbox{ and } \sigma_2 \mbox{ differ on } d \mbox{ variables,  } \sigma_1 \mbox{ falsifies } C_i \mbox{, and }\sigma_2 \mbox{ falsifies } C_j\big\}\big|}_{\alpha(i,j)} \\
& = \Big(2^n - 2\cdot \sum_{i=1}^m 2^{n-|vars(C_i)|}\Big)\cdot \binom{n}{d} + \sum_{i=1}^m \sum_{j=1}^m \alpha(i,j)
\end{align*}
Consider any $1\leq i,j\leq m$. Let us derive an expression for $\alpha(i,j)$. That is, let us count the number of pairs $(\sigma_1, \sigma_2)$ of assignments of $\phi$ such that $\sigma_1$ and $\sigma_2$ differ on $d$ variables, $\sigma_1$ falsifies $C_i$, and $\sigma_2$ falsifies $C_j$. Since $\sigma_1$ falsifies $C_i$, it must set every variable in $vars(C_i)$ such that its corresponding literal in the clause $C_i$ is falsified. That is, for every $x\in vars(C_i)$, if $x$ appears as a positive literal in $C_i$, then $\sigma_1$ must set $x$ to $0$; otherwise, it must set $x$ to $1$. Similarly, since $\sigma_2$ falsifies $C_j$, it must set every variable in $vars(C_j)$ such that its corresponding literal in the clause $C_j$ is falsified. 

There is just one choice for the truth values assigned to the variables in $vars(C_i)\cap vars(C_j)$ by $\sigma_1$ and $\sigma_2$. Also, note that for every variable $x$ in $vars(C_i)\cap vars (C_j)$, if $x$ appears as a positive literal in one clause and as a negative literal in the other clause, then $\sigma_1$ and $\sigma_2$ differ on $x$; otherwise, they overlap on $x$. So, overall, $\sigma_1$ and $\sigma_2$ differ on $\lambda(i,j)$ variables amongst the variables in $vars(C_i)\cap vars(C_j)$. 

Next, we go over all possible choices for the numbers of variables on which $\sigma_1$ and $\sigma_2$ differ (say $d_1, d_2$ and $d_3$ many variables) amongst the variables in $vars(C_i)\setminus vars(C_j), vars(C_j)\setminus vars(C_i)$ and $vars(\phi)\setminus (vars(C_i)\cup vars(C_j))$ respectively. Since $\sigma_1$ and $\sigma_2$ differ on $d$ variables in total, we have $\lambda(i,j)+d_1+d_2+d_3 = d$. 

There is just one choice for the truth values assigned to the variables in $vars(C_i)\setminus vars(C_j)$ by $\sigma_1$, and there are $\binom{|vars(C_i)\setminus vars(C_j)|}{d_1}$ choices for the truth values assigned to the variables in $vars(C_i)\setminus vars(C_j)$ by $\sigma_2$. Similarly, there is just one choice for the truth values assigned to the variables in $vars(C_j)\setminus vars(C_i)$ by $\sigma_2$, and there are $\binom{|vars(C_j)\setminus vars(C_i)|}{d_2}$ choices for the truth values assigned to the variables in $vars(C_j)\setminus vars(C_i)$ by $\sigma_1$. 

There are $\binom{n-|vars(C_i)\cup vars(C_j)|}{d_3}$ choices for the $d_3$ variables on which $\sigma_1$ and $\sigma_2$ differ amongst the variables in $vars(\phi)\setminus (vars(C_i)\cup vars(C_j))$. For each variable $x$ amongst these $d_3$ variables, there are two ways in which $\sigma_1$ and $\sigma_2$ can assign truth values to $x$. That is, either i) $\sigma_1$ sets $x$ to $0$ and $\sigma_2$ sets $x$ to $1$, or ii) $\sigma_1$ sets $x$ to $1$ and $\sigma_2$ sets $x$ to $0$. For each variable $x$ amongst the remaining $n-|vars(C_i)\cup vars(C_j)|-d_3$ variables, there are again two ways in which $\sigma_1$ and $\sigma_2$ can assign truth values to $x$. That is, either i) both $\sigma_1$ and $\sigma_2$ set $x$ to $1$, or ii) both $\sigma_1$ and $\sigma_2$ set $x$ to $0$. So, overall, the number of  ways in which $\sigma_1$ and $\sigma_2$ can assign truth values to the variables in $vars(\phi)\setminus (vars(C_i)\cup vars(C_j))$ is $\binom{n-|vars(C_i)\cup vars(C_j)|}{d_3}\cdot 2^{d_3}\cdot 2^{n-|vars(C_i)\cup vars(C_j)|-d_3}$.

Thus, we get the following expression for $\alpha(i,j)$:
\begin{equation*}
\begin{split}
\resizebox{14.6 cm}{!}{
$2^{n-|vars(C_i)\cup vars(C_j)|} 
\cdot\underset{\substack{d_1, d_2, d_3\geq 0:\\ d_1+d_2+d_3 = d-\lambda(i,j)}}{\sum}\binom{|vars(C_i)\setminus vars(C_j)|}{d_1} \binom{|vars(C_j)\setminus vars(C_i)|}{d_2}\binom{n-|vars(C_i)\cup vars(C_j)|}{d_3}$}
\end{split}
\end{equation*}
Plugging this into the previously obtained equality, we get an expression to count the number of pairs $(\sigma_1, \sigma_2)$ of satisfying assignments of $\phi$ that differ on $d$ variables. Note that this expression can be evaluated in polynomial-time. If the count so obtained is non-zero, we return YES; otherwise, we return NO. So, \textsc{Exact Differ Hitting-SAT} is polynomial-time solvable.  

Note that $(\phi,d)$ is a YES instance of \textsc{Max Differ Hitting-SAT} if and only if at least one of $(\phi, d), (\phi, d+1), \ldots, (\phi, n)$ is a YES instance of \textsc{Exact Differ Hitting-SAT}. Thus, as \textsc{Exact Differ Hitting-SAT} is polynomial-time solvable, so is \textsc{Max Differ Hitting-SAT}. This proves \cref{Exact Differ Hitting-SAT poly-time}.
\end{proof}

\section{Concluding remarks}
In this work, we undertook a complexity-theoretic study of the problem of finding a diverse pair of satisfying assignments of a given Boolean formula, when restricted to affine formulas, $2$-CNF formulas and hitting formulas. This problem can also be studied for i) other classes of formulas on which SAT is polynomial-time solvable,
ii) more than two solutions, and iii) other notions of distance between assignments. Also, its parameterized complexity can be studied in structural parameters of graphs associated with the input Boolean formula; some examples include primal, dual,
 conflict, incidence and consensus graphs \cite{ganian2021new}. An immediate open question is to resolve the parameterized complexity of \textsc{Exact/Max Differ 2-SAT} in the parameter $n-d$\footnote{This question has been recently addressed by Gima et. al. in \cite{gima2024computing}.}. Also, while our polynomial-time algorithm for Hitting formulas counts the number of solution pairs, it is not clear whether it can also be used to output a solution pair\footnote{This came up as a post-talk question from audience during ISAAC 2024.}.

\bibliography{lipics-v2021-sample-article}

\begin{thebibliography}{}

\end{thebibliography}


\begin{thebibliography}{10}

\bibitem{angelsmark2004algorithms}
Ola Angelsmark and Johan Thapper.
\newblock Algorithms for the maximum {H}amming distance problem.
\newblock In {\em International Workshop on Constraint Solving and Constraint Logic Programming}, pages 128--141. Springer, 2004.

\bibitem{aspvall1979linear}
Bengt Aspvall, Michael~F Plass, and Robert~Endre Tarjan.
\newblock A linear-time algorithm for testing the truth of certain quantified boolean formulas.
\newblock {\em Information processing letters}, 8(3):121--123, 1979.

\bibitem{baste2022diversity}
Julien Baste, Michael~R Fellows, Lars Jaffke, Tom{\'a}{\v{s}} Masa{\v{r}}{\'\i}k, Mateus de~Oliveira~Oliveira, Geevarghese Philip, and Frances~A Rosamond.
\newblock Diversity of solutions: An exploration through the lens of fixed-parameter tractability theory.
\newblock {\em Artificial Intelligence}, 303:103644, 2022.

\bibitem{baste2019fpt}
Julien Baste, Lars Jaffke, Tom{\'a}{\v{s}} Masa{\v{r}}{\'\i}k, Geevarghese Philip, and G{\"u}nter Rote.
\newblock {FPT} algorithms for diverse collections of hitting sets.
\newblock {\em Algorithms}, 12(12):254, 2019.

\bibitem{boros1990polynomial}
Endre Boros, Yves Crama, and Peter~L Hammer.
\newblock Polynomial-time inference of all valid implications for horn and related formulae.
\newblock {\em Annals of Mathematics and Artificial Intelligence}, 1:21--32, 1990.

\bibitem{boros1994recognition}
Endre Boros, Peter~L Hammer, and Xiaorong Sun.
\newblock Recognition of q-{H}orn formulae in linear time.
\newblock {\em Discrete Applied Mathematics}, 55(1):1--13, 1994.

\bibitem{conforti1994balanced}
Michele Conforti, G{\'e}rard Cornu{\'e}jols, Ajai Kapoor, Kristina Vu{\v{s}}kovic, and MR~Rao.
\newblock Balanced matrices.
\newblock {\em Mathematical Programming: State of the Art}, pages 1--33, 1994.

\bibitem{cook2023complexity}
Stephen~A Cook.
\newblock The complexity of theorem-proving procedures.
\newblock In {\em Logic, Automata, and Computational Complexity: The Works of Stephen A. Cook}, pages 143--152. 2023.

\bibitem{cygan2015parameterized}
Marek Cygan, Fedor~V Fomin, {\L}ukasz Kowalik, Daniel Lokshtanov, D{\'a}niel Marx, Marcin Pilipczuk, Micha{\l} Pilipczuk, and Saket Saurabh.
\newblock {\em Parameterized algorithms}, volume~5.
\newblock Springer, 2015.

\bibitem{dahllof2005algorithms}
Vilhelm Dahll{\"o}f.
\newblock Algorithms for max {H}amming exact satisfiability.
\newblock In {\em Proceedings of the 16th International Symposium on Algorithms and Computation (ISAAC)}, pages 829--838. Springer, 2005.

\bibitem{de2023finding}
Mark de~Berg, Andr{\'{e}}s~L{\'{o}}pez Mart{\'{\i}}nez, and Frits C.~R. Spieksma.
\newblock Finding diverse minimum s-t cuts.
\newblock In {\em Proceedings of the 34th International Symposium on Algorithms and Computation {ISAAC}}, volume 283 of {\em LIPIcs}, pages 24:1--24:17, 2023.

\bibitem{dowling1984linear}
William~F Dowling and Jean~H Gallier.
\newblock Linear-time algorithms for testing the satisfiability of propositional {H}orn formulae.
\newblock {\em The Journal of Logic Programming}, 1(3):267--284, 1984.

\bibitem{downey1999parametrized}
Rod~G Downey, Michael~R Fellows, Alexander Vardy, and Geoff Whittle.
\newblock The parametrized complexity of some fundamental problems in coding theory.
\newblock {\em SIAM Journal on Computing}, 29(2):545--570, 1999.

\bibitem{downey1992fixed}
Rodney~G Downey and Michael~R Fellows.
\newblock Fixed-parameter intractability.
\newblock In {\em 1992 Seventh Annual Structure in Complexity Theory Conference}, pages 36--37. IEEE Computer Society, 1992.

\bibitem{downey2013fundamentals}
Rodney~G Downey, Michael~R Fellows, et~al.
\newblock {\em Fundamentals of Parameterized Complexity}, volume~4.
\newblock Springer, 2013.

\bibitem{even1975complexity}
S.~Even, A.~Shamir, and A.~Itai.
\newblock On the complexity of time table and multi-commodity flow problems.
\newblock In {\em Proceedings of the 16th Annual Symposium on Foundations of Computer Science (SFCS)}, pages 184--193. IEEE Computer Society, 1975.

\bibitem{fomin2020diverse}
Fedor~V Fomin, Petr~A Golovach, Lars Jaffke, Geevarghese Philip, and Danil Sagunov.
\newblock Diverse pairs of matchings.
\newblock In {\em 31st International Symposium on Algorithms and Computation (ISAAC 2020)}. Schloss Dagstuhl-Leibniz-Zentrum f{\"u}r Informatik, 2020.

\bibitem{fomin2023diverse}
Fedor~V Fomin, Petr~A Golovach, Fahad Panolan, Geevarghese Philip, and Saket Saurabh.
\newblock Diverse collections in matroids and graphs.
\newblock {\em Mathematical Programming}, pages 1--33, 2023.

\bibitem{franco2003perspective}
John Franco and Allen Van~Gelder.
\newblock A perspective on certain polynomial-time solvable classes of satisfiability.
\newblock {\em Discrete Applied Mathematics}, 125(2-3):177--214, 2003.

\bibitem{ganesh2021disjoint}
Aadityan Ganesh, HV~Vishwa~Prakash, Prajakta Nimbhorkar, and Geevarghese Philip.
\newblock Disjoint stable matchings in linear time.
\newblock In {\em Proceedings of the 47th International Workshop on Graph-Theoretic Concepts in Computer Science (WG)}, pages 94--105. Springer, 2021.

\bibitem{ganian2021new}
Robert Ganian and Stefan Szeider.
\newblock New width parameters for sat and\# sat.
\newblock {\em Artificial Intelligence}, 295:103460, 2021.

\bibitem{gima2024computing}
Tatsuya Gima, Yuni Iwamasa, Yasuaki Kobayashi, Kazuhiro Kurita, Yota Otachi, and Rin Saito.
\newblock Computing diverse pair of solutions for tractable sat.
\newblock {\em arXiv preprint arXiv:2412.04016}, 2024.

\bibitem{grcar2011ordinary}
Joseph~F Grcar.
\newblock How ordinary elimination became {G}aussian elimination.
\newblock {\em Historia Mathematica}, 38(2):163--218, 2011.

\bibitem{hanaka2021finding}
Tesshu Hanaka, Yasuaki Kobayashi, Kazuhiro Kurita, and Yota Otachi.
\newblock Finding diverse trees, paths, and more.
\newblock In {\em Proceedings of the AAAI Conference on Artificial Intelligence}, volume~35, pages 3778--3786, 2021.

\bibitem{hoi2019fast}
Gordon Hoi, Sanjay Jain, and Frank Stephan.
\newblock A fast exponential time algorithm for max hamming distance x3sat.
\newblock {\em arXiv preprint arXiv:1910.01293}, 2019.

\bibitem{hoi2019measure}
Gordon Hoi and Frank Stephan.
\newblock Measure and conquer for max {H}amming distance {XSAT}.
\newblock In {\em 30th International Symposium on Algorithms and Computation (ISAAC 2019)}. Schloss Dagstuhl-Leibniz-Zentrum fuer Informatik, 2019.

\bibitem{impagliazzo2001complexity}
Russell Impagliazzo and Ramamohan Paturi.
\newblock On the complexity of k-sat.
\newblock {\em Journal of Computer and System Sciences (JCSS)}, 62(2):367--375, 2001.

\bibitem{iwama1989cnf}
Kazuo Iwama.
\newblock {CNF}-satisfiability test by counting and polynomial average time.
\newblock {\em SIAM Journal on Computing}, 18(2):385--391, 1989.

\bibitem{karp2021reducibility}
Richard Karp.
\newblock Reducibility among combinatorial problems (1972).
\newblock 2021.

\bibitem{koiliaris2019faster}
Konstantinos Koiliaris and Chao Xu.
\newblock Faster pseudopolynomial time algorithms for subset sum.
\newblock {\em ACM Transactions on Algorithms (TALG)}, 15(3):1--20, 2019.

\bibitem{krom1967decision}
Melven~R Krom.
\newblock The decision problem for a class of first-order formulas in which all disjunctions are binary.
\newblock {\em Mathematical Logic Quarterly}, 13(1-2):15--20, 1967.

\bibitem{levin1973universal}
Leonid~Anatolevich Levin.
\newblock Universal sequential search problems.
\newblock {\em Problemy peredachi informatsii}, 9(3):115--116, 1973.

\bibitem{lewis1978renaming}
Harry~R Lewis.
\newblock Renaming a set of clauses as a {H}orn set.
\newblock {\em Journal of the ACM (JACM)}, 25(1):134--135, 1978.

\bibitem{mironov2006applications}
Ilya Mironov and Lintao Zhang.
\newblock Applications of sat solvers to cryptanalysis of hash functions.
\newblock In {\em Proceedings of the 9th International Conference on Theory and Applications of Satisfiability Testing (SAT)}, pages 102--115. Springer, 2006.

\bibitem{misra2022diverse}
Neeldhara Misra, Harshil Mittal, and Saraswati Nanoti.
\newblock Diverse non crossing matchings.
\newblock In {\em CCCG}, pages 249--256, 2022.

\bibitem{misra2024parameterized}
Neeldhara Misra, Harshil Mittal, and Ashutosh Rai.
\newblock On the parameterized complexity of diverse sat.
\newblock In {\em 35th International Symposium on Algorithms and Computation (ISAAC 2024)}, pages 50--1. Schloss Dagstuhl--Leibniz-Zentrum f{\"u}r Informatik, 2024.

\bibitem{mohar2001face}
Bojan Mohar.
\newblock Face covers and the genus problem for apex graphs.
\newblock {\em Journal of Combinatorial Theory (JCT), Series B}, 82(1):102--117, 2001.

\bibitem{schaefer1978complexity}
Thomas~J Schaefer.
\newblock The complexity of satisfiability problems.
\newblock In {\em Proceedings of the tenth annual ACM Symposium on Theory of Computing}, pages 216--226, 1978.

\bibitem{scutella1990note}
Maria~Grazia Scutella.
\newblock A note on {D}owling and {G}allier's top-down algorithm for propositional horn satisfiability.
\newblock {\em The Journal of Logic Programming}, 8(3):265--273, 1990.

\bibitem{vizel2015boolean}
Yakir Vizel, Georg Weissenbacher, and Sharad Malik.
\newblock Boolean satisfiability solvers and their applications in model checking.
\newblock {\em Proceedings of the IEEE}, 103(11):2021--2035, 2015.

\end{thebibliography}

\end{document}